\documentclass[DIV13,10pt]{scrartcl}

\usepackage[noblocks]{authblk}

\usepackage{enumerate}

\setkomafont{title}{\normalfont \LARGE \bfseries}
\setkomafont{section}{\normalfont \Large \bfseries \boldmath}
\setkomafont{subsection}{\normalfont \large \bfseries}
\setkomafont{subsubsection}{\normalfont \normalsize \bfseries}
\setkomafont{descriptionlabel}{\itshape}
\setkomafont{paragraph}{\normalfont \bfseries \boldmath}

\usepackage[OT1]{fontenc}
\usepackage[english]{babel}
\usepackage[utf8]{inputenc}
\usepackage{charter}
\usepackage[svgnames]{xcolor}
\usepackage{amssymb,amsfonts,mathrsfs,yhmath}
\usepackage[amsthm,amsmath,thmmarks]{ntheorem}
\usepackage{enumitem}
\usepackage{expdlist,graphicx}
\usepackage{booktabs}
\usepackage{caption}
\usepackage[labelformat=simple]{subcaption}

\usepackage{epsfig}
\usepackage[numbers]{natbib}
\usepackage[linkcolor=Blue,colorlinks=true,citecolor=ForestGreen,filecolor=red,urlcolor=Blue]{hyperref}
\usepackage{cleveref}[2011/12/24]
\usepackage{enumitem,linegoal}

\usepackage{tikz}
\usetikzlibrary{arrows}
\usetikzlibrary{arrows.meta}
\usetikzlibrary{decorations.pathmorphing}
\usetikzlibrary{decorations.markings}
\usetikzlibrary{calc}

\tikzset{
  cir/.style = {circle,draw,fill,inner sep=.7pt},
  circ/.style = {circle,draw,fill,inner sep=1.3pt},
  circg/.style = {circle,draw=lightgray,fill=lightgray,inner sep=1.3pt},
  circr/.style = {circle,draw=red,fill=red,inner sep=1.3pt},
  invisible/.style = {draw=none,inner sep=0pt,font=\tiny},
  nonedge/.style={decorate,decoration={snake,amplitude=.3mm,segment length=1mm},draw}
}

\theoremstyle{plain}
\newtheorem{theorem}{Theorem}
\newtheorem{corollary}[theorem]{Corollary}
\newtheorem{lemma}[theorem]{Lemma}

\newtheorem{observation}[theorem]{Observation}
\newtheorem{claim}{Claim}
\theoremstyle{definition}
\newtheorem{definition}{Definition}

\newtheorem{question}[theorem]{Question}
\theoremstyle{remark}

\usepackage[pagewise,displaymath,mathlines]{lineno}

\makeatletter
\newcommand{\eqnum}{\leavevmode\hfill\refstepcounter{equation}\textup{\tagform@{\theequation}}}
\makeatother

\date{}

\title{\Large Semitotal Domination: New hardness results and a polynomial-time algorithm for graphs of bounded mim-width}
\author{\normalsize Esther Galby, Andrea Munaro, Bernard Ries\\Department of Informatics, University of Fribourg, Bd de P\'erolles 90, 1700 Fribourg, Switzerland}

\begin{document}
\maketitle

\begin{abstract}
\textbf{Abstract.} A semitotal dominating set of a graph $G$ with no isolated vertex is a dominating set $D$ of $G$ such that every vertex in $D$ is within distance two of another vertex in $D$. The minimum size $\gamma_{t2}(G)$ of a semitotal dominating set of $G$ is squeezed between the domination number $\gamma(G)$ and the total domination number~$\gamma_{t}(G)$. 

\textsc{Semitotal Dominating Set} is the problem of finding, given a graph $G$, a semitotal dominating set of $G$ of size $\gamma_{t2}(G)$. In this paper, we continue the systematic study on the computational complexity of this problem when restricted to special graph classes. In particular, we show that it is solvable in polynomial time for the class of graphs with bounded mim-width by a reduction to \textsc{Total Dominating Set} and we provide several approximation lower bounds for subclasses of subcubic graphs. Moreover, we obtain complexity dichotomies in monogenic classes for the decision versions of \textsc{Semitotal Dominating Set} and \textsc{Total Dominating Set}. 

Finally, we show that it is $\mathsf{NP}$-complete to recognise the graphs such that $\gamma_{t2}(G) = \gamma_{t}(G)$ and those such that $\gamma(G) = \gamma_{t2}(G)$, even if restricted to be planar and with maximum degree at most $4$, and we provide forbidden induced subgraph characterisations for the graphs heriditarily satisfying either of these two equalities. 
\end{abstract}


\section{Introduction}

Two of the most studied parameters in graph theory are the domination number and the total domination number. Recall that a \textit{dominating set} of a graph $G$ is a subset $D \subseteq V(G)$ such that each vertex in $V(G) \setminus D$ is adjacent to a vertex in $D$. The \textit{domination number} $\gamma(G)$ is the minimum size of a dominating set of $G$. 
A \textit{total dominating set} (TD-set for short) of a graph $G$ with no isolated vertex is a dominating set $D$ of $G$ such that every vertex in $D$ is adjacent to another vertex in $D$. The \textit{total domination number} $\gamma_{t}(G)$ is the minimum size of a TD-set of $G$. Dominating sets and total dominating sets are the subjects of several monographs \citep{HHS98a,HHS98b,HY13}. 

The relaxed notion of semitotal dominating set was introduced by \citet{GHM14}. A \textit{semitotal dominating set} (semi-TD-set for short) of a graph $G$ with no isolated vertex is a dominating set $D$ of $G$ such that every vertex in $D$ is within distance two of another vertex in $D$. In this case, we say that every vertex has a witness: given a dominating set $D$ of $G$ and $v \in D$, a \textit{witness for $v$ with respect to $D$} is a vertex $w \in D$ such that $d(v, w) \leq 2$. In other words, a dominating set $D$ is a semi-TD-set if and only if every vertex in $D$ has a witness with respect to $D$. The \textit{semitotal domination number} $\gamma_{t2}(G)$ is the minimum size of a semi-TD-set of $G$. It immediately follows from the definitions that if $G$ is a graph with no isolated vertex, then $$\gamma(G) \leq \gamma_{t2}(G) \leq \gamma_{t}(G).$$

\citet{GHM14} showed that \textsc{Semitotal Dominating Set} is $\mathsf{NP}$-complete. This is the problem of deciding, given a graph $G$ and an integer $k$, whether $\gamma_{t2}(G) \leq k$. A systematic study on the computational complexity of \textsc{Semitotal Dominating Set} was initiated by \citet{HP17}. They showed that the problem is $\mathsf{NP}$-complete when restricted to the following graph classes: bipartite graphs, chordal bipartite graphs, planar graphs and split graphs. They also considered the minimisation version of \textsc{Semitotal Dominating Set}. Note that in this paper we denote the decision version and the minimisation version of a given problem by the same name. They showed that \textsc{Semitotal Dominating Set} has the same approximation hardness as the well-known \textsc{Set Cover}: it is not approximable within $(1 - \varepsilon)\ln n$, for any $\varepsilon > 0$, unless $\mathsf{NP} \subseteq \mathsf{DTIME}(n^{O(\log\log n)})$. On the other hand, using the natural greedy algorithm for \textsc{Set Cover}, they showed that it is in $\mathsf{APX}$ for graphs with bounded degree. Moreover, it is \textsc{APX}-complete for bipartite graphs with maximum degree $4$ and, on the positive side, they showed that it is solvable in $O(n^{2})$ time for interval graphs.  

In \Cref{secapp}, we provide the following approximation lower bounds for \textsc{Semitotal Dominating Set}. It is not approximable within $1.00013956$, unless $\mathsf{P} = \mathsf{NP}$, even when restricted to subcubic line graphs of bipartite graphs. Moreover, it is $\mathsf{APX}$-complete when restricted to cubic graphs and to subcubic bipartite graphs. The latter result is obtained by studying the effect of odd subdivisions on the semitotal domination number. Our results answer a question by \citet{HP17} on the complexity of \textsc{Semitotal Dominating Set} restricted to subcubic graphs.

In \Cref{reduc}, we introduce a graph transformation allowing to reduce \textsc{Semitotal Dominating Set} to \textsc{Total Dominating Set}. More specifically, from a given graph $G$, we construct the transformed graph $G'$ by first adding a true twin for each vertex of $G$ and then adding edges between vertices of $G$ at distance $2$ (in particular, $G'$ contains a copy of $G^2$). We then show that $\gamma_{t2}(G) = \gamma_t(G')$ and that a minimum semi-TD-set of $G$ can be obtained from a minimum TD-set of $G'$ in linear time. 

The class of graphs with bounded maximum induced matching width (mim-width for short) appears to behave well with respect to some domination problems, such as \textsc{Dominating Set} and \textsc{Total Dominating Set}. Mim-width is a parameter introduced by \citet{Vat12} and measuring how easy it is to decompose a graph along vertex cuts inducing a bipartite graph with small maximum induced matching size. Combining results in \citep{BV13,BTV13}, it is known that $(\sigma, \rho)$-domination problems (a class of graph problems including \textsc{Dominating Set} and \textsc{Total Dominating Set} introduced by \citet{TP97}) can be solved in $O(n^{w})$ time, assuming a decomposition tree with mim-width $w$ is provided as part of the input. Even though deciding the mim-width of a graph is $\mathsf{NP}$-hard in general and not in $\mathsf{APX}$ unless $\mathsf{NP} = \mathsf{ZPP}$ \citep{SV16}, \citet{BV13} showed that it is possible to find decomposition trees of constant mim-width in polynomial time for the following classes of graphs: permutation graphs, convex graphs and their complements, interval graphs and their complements, (circular $k$-) trapezoid graphs, circular permutation graphs, Dilworth-$k$ graphs, $k$-polygon graphs, circular arc graphs and complements of $d$-degenerate graphs. Another class allowing polynomial-time algorithms for some domination problems is that of dually chordal graphs, a superclass of interval graphs \citep{BCD98}. However, it should be remarked that dually chordal graphs do not have bounded mim-width \citep{Men17}.

In \Cref{mim}, we provide a polynomial-time algorithm for \textsc{Semitotal Dominating Set} restricted to graphs with bounded mim-width. In view of the reduction from \Cref{reduc}, we show that graphs with bounded mim-width are closed under the transformation therein and we rely on polynomial-time algorithms for \textsc{Total Dominating Set} restricted to graphs with bounded mim-width. Our result answers two questions in \citep{HP17} on the complexity of \textsc{Semitotal Dominating Set} for bipartite permutation graphs and convex bipartite graphs, two subclasses of chordal bipartite graphs having bounded mim-width. 

In \Cref{dually}, we show that also the class of dually chordal graphs is closed under the transformation above. Unfortunately, in this case, we cannot rely on any polynomial-time algorithm for \textsc{Total Dominating Set} restricted to this class (the algorithm provided in \citep{KS97} does not seem to be correct) and we leave the determination of the complexity status of \textsc{Semitotal Dominating Set} and \textsc{Total Dominating Set} for dually chordal graphs as an open problem. 

A class of graphs $\mathcal{G}$ is \textit{monogenic} if it is defined by a single forbidden induced subgraph, i.e. $\mathcal{G} = \mbox{\textit{Free}}(H)$, for some graph $H$. \citet{Kor92} showed that \textsc{Dominating Set} is decidable in polynomial time if $H$ is an induced subgraph of $P_{4} + tK_{1}$, for $t \geq 0$, and $\mathsf{NP}$-complete otherwise. In \Cref{sec:dicho}, combining results of \Cref{secapp} with existing ones, we obtain complexity dichotomies in monogenic classes for \textsc{Semitotal Dominating Set} and \textsc{Total Dominating Set}. It turns out that the complexities of \textsc{Dominating Set}, \textsc{Semitotal Dominating Set}, \textsc{Total Dominating Set} and \textsc{Connected Dominating Set} all agree when restricted to monogenic classes. 

\begin{center}
\begin{tabular}{l*{3}{l}}
Graph class                       & \textsc{Dominating Set} & \textsc{Semitotal Dominating Set} & \textsc{Total Dominating Set}  \\
\hline
bipartite                         & $\mathsf{NP}$-c \citep{Ber84} & $\mathsf{NP}$-c \citep{HP17} & $\mathsf{NP}$-c \citep{PLH83}   \\
line graph of bipartite           & $\mathsf{NP}$-c \citep{Kor92} & $\mathsf{NP}$-c (\Cref{semidec}) & $\mathsf{NP}$-c \citep{Mc94}  \\
circle                            & $\mathsf{NP}$-c \citep{Keil93} & \textbf{\color{red}?} & $\mathsf{NP}$-c \citep{Keil93}     \\
split                             & $\mathsf{NP}$-c \citep{Ber84} & $\mathsf{NP}$-c \citep{HP17} & $\mathsf{NP}$-c \citep{LP83}   \\
chordal bipartite                 & $\mathsf{NP}$-c \citep{MB87} & $\mathsf{NP}$-c \citep{HP17} & $\mathsf{P}$ \citep{DMK90}  \\
dually chordal                    & $\mathsf{P}$ \citep{BCD98} & \textbf{\color{red}?} & \textbf{\color{red}?}      \\
AT-free                           & $\mathsf{P}$ \citep{Kra00} & \textbf{\color{red}?} & $\mathsf{P}$ \citep{Kra00} \\ 
bounded clique-width              & $\mathsf{P}$ \citep{Cou90} & $\mathsf{P}$ \citep{Cou90}   & $\mathsf{P}$ \citep{Cou90}    \\
bounded mim-width                 & $\mathsf{P}$ \citep{BV13,BTV13} & $\mathsf{P}$ (\Cref{boundmimsemi}) & $\mathsf{P}$ \citep{BV13,BTV13}    \\  
tolerance                         & $\mathsf{P}$ \citep{GM16} & \textbf{\color{red}?} & \textbf{\color{red}?}     \\  
\end{tabular}
\captionof{table}{Comparison between the computational complexities of \textsc{Dominating Set}, \textsc{Semitotal Dominating Set} and \textsc{Total Dominating Set}.}
\label{table}
\end{center}

\begin{figure}[h]
\centering
\begin{tikzpicture}[scale=.95]
\node[rectangle,draw] (perfect) at (6,16) {perfect};
\node[rectangle,draw] (line) at (-1,15.2) {\parbox{2cm}{\centering line graph\\ of bipartite}};
\node[rectangle,draw] (comp) at (2,15) {comparability};
\node[rectangle,draw] (tolerance) at (4.5,15) {tolerance \textbf{\color{red}?}};
\node[rectangle,draw] (cocomp) at (8,14) {cocomparability \textbf{\color{red}?}};
\node[rectangle,draw] (chordal) at (11,15) {chordal};
\node[rectangle,draw] (atfree) at (8.2,15) {AT-free \textbf{\color{red}?}};
\node[rectangle,draw] (bipartite) at (-1,14) {bipartite};
\node[rectangle,draw] (chordalbi) at (-1,13) {chordal bipartite};
\node[rectangle,draw] (convex) at (-1,12) {convex};
\node[rectangle,draw] (circperm) at (2,12) {circular permutation};
\node[rectangle,draw] (perm) at (2,11) {permutation};
\node[rectangle,draw] (permbi) at (0,10) {bipartite permutation};
\node[rectangle,draw] (kpoly) at (5,12) {$k$-polygon};
\node[rectangle,draw] (circle) at (5,13) {circle \textbf{\color{red}?}};
\node[rectangle,draw] (circktrap) at (9,12) {circular $k$-trapezoid};
\node[rectangle,draw] (ktrap) at (9,11) {$k$-trapezoid};
\node[rectangle,draw] (trap) at (9,10) {trapezoid};
\node[rectangle,draw] (boundtol) at (8,9) {bounded tolerance};
\node[rectangle,draw] (disthere) at (4,9) {distance hereditary};
\node[rectangle,draw] (cograph) at (4,8) {cograph};
\node[rectangle,draw] (threshold) at (4,7) {threshold};
\node[rectangle,draw] (dilworth) at (7,8) {Dilworth $k$};
\node[rectangle,draw] (circarc) at (11.5,10) {circular arc};
\node[rectangle,draw] (interval) at (11,8) {interval};
\node[rectangle,draw] (dually) at (12.5,12) {dually chordal \textbf{\color{red}?}};
\node[rectangle,draw] (split) at (14.3,13) {split};
\node[rectangle,draw] (strong) at (12.5,11) {strongly chordal};
\node[rectangle,draw] (ktree) at (13,6) {$k$-tree, fixed $k$};
\node[rectangle,draw] (tree) at (2,5) {tree};

\draw[-] (perfect) -- (line) 
(perfect) -- (comp)
(perfect) -- (tolerance)
(perfect) -- (cocomp)
(perfect) -- (chordal)
(perfect) ..controls (6.5,11).. (disthere)
(cocomp) -- (atfree)
(cocomp) ..controls (6.7,11.5).. (ktrap)
(comp) -- (bipartite)
(comp) -- (circperm)
(bipartite) -- (chordalbi)
(chordalbi) -- (convex) 
(chordalbi) ..controls (-4,10).. (tree)
(circperm) -- (perm)
(perm) -- (permbi)
(convex) -- (permbi)
(perm) -- (boundtol)
(perm) -- (kpoly)
(perm) -- (boundtol)
(kpoly)-- (circle)
(disthere) -- (cograph)
(cograph) -- (threshold)
(threshold) -- (dilworth)
(threshold) -- (interval)
(circktrap) -- (ktrap) 
(ktrap) -- (trap)
(trap) -- (boundtol)
(circarc) -- (interval)
(chordal) ..controls (10.7,11.6).. (strong)
(dually) -- (strong)
(chordal) ..controls (14,14).. (split)
(strong) ..controls (13,10).. (interval)
(circarc) ..controls (10.7,11).. (circktrap)
(tolerance) ..controls (6.2,13).. (boundtol)
(boundtol) -- (interval)
(split) ..controls (13.6,7).. (threshold)
(tree) ..controls (2.5,8).. (disthere)
(tree) -- (ktree)
(chordal) ..controls (16,14).. (ktree);

\draw[dashed,blue,thick] (-3.5,12.5) -- (15.3,12.5) node[pos=.01,above] {$\mathsf{NP}$-c} node[pos=0,below] {$\mathsf{P}$};
\end{tikzpicture}
\caption{Computational complexity of \textsc{Semitotal Dominating Set}. The position of the graph classes whose complexity is unknown (red question mark) with respect to the boundary $\mathsf{P}$/$\mathsf{NP}$-c is completely arbitrary.} 
\label{fig:dia}
\end{figure}

Since $\gamma(G) \leq \gamma_{t2}(G) \leq \gamma_{t}(G)$, for any graph $G$ with no isolated vertex, it is natural to ask for which graphs equalities hold. In \Cref{sec:perfect}, we show that the graphs attaining either of the two equalities are unlikely to have a polynomial characterisation. Indeed, it is $\mathsf{NP}$-complete to recognise the graphs such that $\gamma_{t2}(G) = \gamma_{t}(G)$ and those such that $\gamma(G) = \gamma_{t2}(G)$, even if restricted to be planar and with maximum degree at most $4$. Therefore, it makes sense to study the class of graphs obtained by further requiring equality for every induced subgraph. Graphs for which the equality between certain domination parameters holds for every induced subgraph received considerable attention (see, e.g., \citep{ADR15,CP17,HH18,Sch12,Zve03}).
We provide forbidden induced subgraph characterisations for the graphs heriditarily satisfying $\gamma_{t2} = \gamma_{t}$ and $\gamma = \widetilde{\gamma}_{t2}$, where 
\begin{equation*}
  \widetilde{\gamma}_{t2}(G)=
     \begin{cases}
        1 & \text{if $G$ has a dominating vertex, i.e. a vertex adjacent to all the other vertices;} \\
        \gamma_{t2}(G) & \text{otherwise.}
     \end{cases}
\end{equation*}
This variation in the definition of $\gamma_{t2}$ is motivated by the fact that if $G$ has no isolated vertex and contains a dominating vertex, then $\gamma(G) = 1$ but $\gamma_{t2}(G) = 2$. Our characterisations have been independently obtained by \citet{HH18} but the proofs we provide have the nice feature of being considerably shorter.

\section{Preliminaries}\label{prel}

In this paper we consider only finite simple graphs. Given a graph $G$, we usually denote its vertex set by $V(G)$ and its edge set by $E(G)$. 

\vspace{0.3cm}

\textbf{Neighbourhoods and degrees.} For a vertex $v \in V(G)$, the \textit{neighbourhood} $N_{G}(v)$ is the set of vertices adjacent to $v$ in $G$. The \textit{degree} $d_{G}(v)$ of a vertex $v \in V(G)$ is the number of edges incident to $v$ in $G$. A \textit{$k$-vertex} is a vertex of degree $k$. We refer to a $3$-vertex as a \textit{cubic} vertex and to a $0$-vertex as an \textit{isolated} vertex. The \textit{maximum degree} $\Delta(G)$ of $G$ is the quantity $\max\left\{d_{G}(v): v \in V\right\}$ and $G$ is \textit{subcubic} if $\Delta(G) \leq 3$. If all the vertices of $G$ have the same degree $k$, then $G$ is \textit{$k$-regular} and a \textit{cubic} graph is a $3$-regular graph.  

Two subsets $X$ and $Y$ of $V(G)$ with $X \cap Y = \varnothing$ are \textit{complete} to each other if every vertex of $X$ is adjacent to every vertex of $Y$. 

\vspace{0.3cm}

\textbf{Paths and cycles.} A \textit{path} is a non-empty graph $P = (V, E)$ with $V = \left\{x_{0}, x_{1}, \dots, x_{k}\right\}$ and $E = \left\{x_{0}x_{1}, x_{1}x_{2}, \dots, x_{k-1}x_{k}\right\}$, and where the $x_{i}$'s are all distinct. The vertices $x_{0}$ and $x_{k}$ are \textit{linked} by $P$ and they are called the \textit{ends} of $P$. The \textit{length} of a path is its number of edges and the path on $n$ vertices is denoted by $P_{n}$. We refer to a path $P$ by a natural sequence of its vertices: $P = x_{0}x_{1}\cdots x_{k}$. Such a path $P$ is a path \textit{between $x_{0}$ and $x_{k}$}, or a \textit{$x_{0}, x_{k}$-path}. If $P = x_{0}\cdots x_{k}$ is a path and $k \geq 2$, the graph with vertex set $V(P)$ and edge set $E(P) \cup \left\{x_{k}x_{0}\right\}$ is a \textit{cycle}. The cycle on $n$ vertices is denoted by $C_{n}$.  

The \textit{distance} $d_{G}(u, v)$ from a vertex $u$ to a vertex $v$ in a graph $G$ is the length of a shortest path between $u$ and $v$. If $u$ and $v$ are not linked by any path in $G$, we set $d_{G}(u, v) = \infty$. Given two subsets $X$ and $Y$ of $V(G)$, the \textit{distance} from $X$ to $Y$ is the quantity $d_{G}(X, Y) = \min_{x \in X, y \in Y} d_{G}(x, y)$, i.e. it is the minimum length of a path between a vertex in $X$ and a vertex in $Y$. The \textit{disk centered at $z$ of radius $r$} is the set $N^r[z] = \{v \in V(G) : d_{G}(z, v) \leq r\}$. 

The \textit{girth} of a graph containing a cycle is the length of a shortest cycle and a graph with no cycle has infinite girth.

\vspace{0.3cm}

\textbf{Graph operations.} Given a graph $G = (V, E)$ and $V' \subseteq V$, the operation of \textit{deleting the set of vertices $V'$} from $G$ results in the graph $G - V' = G[V\setminus V']$. The \textit{union} of simple graphs $G$ and $H$ is the graph $G \cup H$ with vertex set $V(G) \cup V(H)$ and edge set $E(G) \cup E(H)$. If $G$ and $H$ are disjoint, their union is denoted by $G + H$ and the union of $k$ (disjoint) copies of $G$ is denoted by $kG$. A \textit{$k$-subdivision} of $G$ is the graph obtained from $G$ by adding $k$ new vertices for each edge of $G$, i.e. each edge is replaced by a path of length $k + 1$. 
     
\vspace{0.3cm}  

\textbf{Graph classes and special graphs.} If a graph does not contain induced subgraphs isomorphic to graphs in a set $Z$, it is \textit{$Z$-free} and the set of all $Z$-free graphs is denoted by $\mbox{\textit{Free}}(Z)$. A class of graphs is \textit{hereditary} if it is closed under deletions of vertices. It is well-known and easy to see that a class of graphs $X$ is hereditary if and only if it can be defined by a set of forbidden induced subgraphs, i.e. $X = \mbox{\textit{Free}}(Z)$ for some set of graphs $Z$. 

A \textit{complete graph} is a graph whose vertices are pairwise adjacent and the complete graph on $n$ vertices is denoted by $K_{n}$. A \textit{triangle} is the graph $K_{3}$. A graph $G$ is \textit{$r$-partite}, for $r \geq 2$, if its vertex set admits a partition into $r$ classes such that every edge has its endpoints in different classes. $2$-partite graphs are usually called \textit{bipartite}. A \textit{split graph} is a graph whose vertices can be partitioned into a clique and an independent set (see below for the definitions of clique and independent set). The minimal set of forbidden induced subgraphs for the class of split graphs is $\{2K_{2}, C_{4}, C_{5}\}$. A \textit{cograph} is defined recursively as follows: $K_{1}$ is a cograph, the disjoint union of cographs is a cograph, the complement of a cograph is a cograph. The class of cographs coincides with that of $P_{4}$-free graphs.

\vspace{0.3cm}
           
\textbf{Graph parameters.} A set of vertices or edges of a graph is \textit{minimum} with respect to the property $\mathcal{P}$ if it has minimum size among all subsets having property $\mathcal{P}$. The term \textit{maximum} is defined analogously. In this paper we often consider the following parameters of a graph $G$.

An \textit{independent set} of a graph is a set of pairwise non-adjacent vertices. A \textit{clique} of a graph is a set of pairwise adjacent vertices. A \textit{matching} of a graph is a set of pairwise non-incident edges. A \textit{vertex cover} of a graph is a subset of vertices containing at least one endpoint of every edge. The size of a minimum vertex cover of $G$ is denoted by $\beta(G)$. Clearly, $S \subseteq V(G)$ is a vertex cover of $G$ if and only if $V(G)\setminus S$ is an independent set of $G$. 

\vspace{0.3cm}

\textbf{Monadic second-order logic of graphs.} We refer to \citep{KLM09} for the definitions of tree-width and clique-width. Graphs of bounded tree-width are particularly interesting from an algorithmic point of view: many $\mathsf{NP}$-complete problems can be solved in linear time for them. A celebrated algorithmic meta-theorem of \citet{Cou90} provides a way to quickly establish that a certain problem is decidable in linear time on graphs of bounded tree-width: all (graph) problems expressible in monadic second-order logic with edge-set quantification are decidable in linear time on graphs of bounded tree-width, assuming a tree decomposition is given (see also \citep{ALS91}). Moreover, for a fixed $k$, it is in fact possible to test in linear time whether a graph has tree-width at most $k$ and, if so, to find a tree-decomposition with width at most $k$ \citep{Bod96}. 

Let us briefly recall that monadic second-order logic is an extension of first-order logic by quantification over sets. The language of \textit{monadic second-order logic of graphs} ($\mbox{MSO}_{1}$ for short) contains the expressions built from the following elements:
\begin{itemize}
\item Variables $x, y, \dots$ for vertices and $X, Y, \dots$ for sets of vertices;
\item Predicates $x \in X$ and $\textit{adj}(x, y)$;
\item Equality for variables, standard Boolean connectives and the quantifiers $\forall$ and $\exists$.
\end{itemize}
By considering edges and sets of edges as other sorts of variables and the incidence predicate $\textit{inc}(v, e)$, we obtain \textit{monadic second-order logic of graphs with edge-set quantification} ($\mbox{MSO}_{2}$ for short).

A notion related to tree-width is that of clique-width. As shown by \citet{CO00}, every graph of bounded tree-width has bounded clique-width but there are graphs of bounded clique-width having unbounded tree-width (for example, complete graphs). Therefore, clique-width can be viewed as a more general concept than tree-width. An important class of graphs having bounded clique-width is that of cographs: they have clique-width at most $2$. 

Similarly to tree-width, having bounded clique-width has interesting algorithmic implications. If a graph property is expressible in the more restricted $\mbox{MSO}_{1}$, then \citet{CMR00} showed that it is decidable in linear time even for graphs of bounded clique-width, assuming a $k$-expression of the graph is explicitly given. On the other hand, for fixed $k$, \citet{OS06} provided a polynomial-time algorithm that given a graph $G$ either decides $G$ has clique-width at least $k + 1$ or outputs a $2^{3k + 2} - 1$-expression. Therefore, a graph property expressible in $\mbox{MSO}_{1}$ is decidable in polynomial time for graphs of bounded clique-width. 

\vspace{0.3cm}

\textbf{Approximation hardness.} We refer the reader to \citep{Du12,Vaz01} for introductions to approximation hardness. Recall that an \textit{optimisation problem} $\Pi$ is a quadruple $(\mathcal{I}, S, c, \mbox{opt})$, where $\mathcal{I}$ is the set of instances of $\Pi$, $S(I)$ is the set of feasible solutions of an instance $I \in \mathcal{I}$, the function $c\colon \mathcal{I} \times S \rightarrow \mathbb{N}$ is the objective function and $\mbox{opt} \in \{\max, \min\}$. We denote by $\mbox{opt}(I)$ the value $\mbox{opt}\{c(I, s) : s \in S(I)\}$.

Let $\Pi = (\mathcal{I}, S, c, \min)$ and $\Pi' = (\mathcal{I}', S', c', \min)$ be two minimisation problems. We say that $\Pi$ \textit{L-reduces} to $\Pi'$ if there exists a polynomial-time function $f\colon \mathcal{I} \rightarrow \mathcal{I}'$ and positive constants $\alpha$ and $\beta$ such that for every instance $I$ of $\Pi$ the following hold:

\begin{itemize}
\item $\min(f(I)) \leq \alpha\cdot\min(I)$;
\item For every feasible solution $s' \in S'(f(I))$, we can find in polynomial time a feasible solution $s \in S(I)$ such that $|\min(I) - c(I, s)| \leq \beta\cdot|\min(f(I)) - c'(f(I), s')|$.     
\end{itemize}

Let $0 < \alpha < \beta$. A minimisation problem $\Pi = (\mathcal{I}, S, c, \min)$ has an \textit{$\mathsf{NP}$-hard gap} $[\alpha, \beta]$ if there exist an $\mathsf{NP}$-complete decision problem $\Lambda = (\mathcal{I}_{\Lambda}, S_{\Lambda})$ and a polynomial-time reduction $f$ from $\Lambda$ to $\Pi$ such that, for every $I \in \mathcal{I}_{\Lambda}$, the following hold:
\begin{itemize}
\item If $S_{\Lambda}(I) = \mbox{``yes''}$, then $\min(f(I)) \leq \alpha$;
\item If $S_{\Lambda}(I) = \mbox{``no''}$, then $\min(f(I)) > \beta$.
\end{itemize} 

\begin{lemma}[see \citep{Du12}]\label{gapdec} If $\Pi$ is an optimisation problem with an $\mathsf{NP}$-hard gap $[\alpha, \beta]$, for some $0 < \alpha < \beta$, then there is no $\frac{\beta}{\alpha}$-approximation algorithm for $\Pi$, unless $\mathsf{P} = \mathsf{NP}$.
\end{lemma}

A \textit{gap-preserving reduction} from a minimisation problem $\Pi = (\mathcal{I}_{\Pi}, S_{\Pi}, c_{\Pi}, \min)$ to a minimisation problem $\Lambda = (\mathcal{I}_{\Lambda}, S_{\Lambda}, c_{\Lambda}, \min)$ is a function $f$ mapping every instance of $\Pi$ to an instance of $\Lambda$ in polynomial time, together with constants $\alpha_{\Pi} \geq 1$ and $\alpha_{\Lambda} \geq 1$ and functions $g_{\Pi}$ and $g_{\Lambda}$ such that:
\begin{itemize}
\item If $\min(I) \leq g_{\Pi}(I)$, then $\min(f(I)) \leq g_{\Lambda}(f(I))$;
\item If $\min(I) > \alpha_{\Pi}g_{\Pi}(I)$, then $\min(f(I)) > \alpha_{\Lambda}g_{\Lambda}(f(I))$.
\end{itemize}  
Similarly to \Cref{gapdec}, the following holds:

\begin{lemma}[see \citep{Du12}]\label{gapred} Suppose there exists a \textit{gap-preserving reduction} from a minimisation problem $\Pi = (\mathcal{I}_{\Pi}, S_{\Pi}, c_{\Pi}, \min)$ to a minimisation problem $\Lambda = (\mathcal{I}_{\Lambda}, S_{\Lambda}, c_{\Lambda}, \min)$. If it is $\mathsf{NP}$-hard to distinguish between those $I \in \mathcal{I}_{\Pi}$ for which $\min(I) \leq g_{\Pi}(I)$ and those for which $\min(I) > \alpha_{\Pi}g_{\Pi}(I)$, then it is $\mathsf{NP}$-hard to distinguish between those $f(I) \in \mathcal{I}_{\Lambda}$ for which $\min(f(I)) \leq g_{\Lambda}(f(I))$ and those for which $\min(f(I)) > \alpha_{\Lambda}g_{\Lambda}(f(I))$. In particular, $\Lambda$ is not approximable within $\alpha_{\Lambda}$, unless $\mathsf{P} = \mathsf{NP}$.
\end{lemma}

\section{Approximation hardness}\label{secapp}

\citet{HP17} showed that \textsc{Semitotal Dominating Set} has the same approximation hardness as the well-known \textsc{Set Cover}: \textsc{Semitotal Dominating Set} is not approximable within $(1 - \varepsilon)\ln n$, for any $\varepsilon > 0$, unless $\mathsf{NP} \subseteq \mathsf{DTIME}(n^{O(\log\log n)})$. Note that the same inapproximability result holds for \textsc{Dominating Set}, \textsc{Total Dominating Set} and \textsc{Connected Dominating Set} \citep{CC08}. On the other hand, using the natural greedy algorithm for \textsc{Set Cover}, \citet{HP17} showed that \textsc{Semitotal Dominating Set} is in $\mathsf{APX}$ for graphs with bounded degree: it admits a $2 + 3 \ln(\Delta + 1)$ approximation algorithm for graphs with maximum degree $\Delta$. Moreover, they showed that it is \textsc{APX}-complete for bipartite graphs with maximum degree $4$ and asked whether the same holds for subcubic graphs.

In this section, we answer their question in the affirmative and provide approximation lower bounds for \textsc{Semitotal Dominating Set} when restricted to some subclasses of subcubic graphs. We begin by showing the following inapproximability result:  

\begin{theorem}\label{inapprox} \textsc{Semitotal Dominating Set} is not approximable within $1.00013956$, unless $\mathsf{P} = \mathsf{NP}$, even when restricted to subcubic line graphs of bipartite graphs.
\end{theorem}

\begin{proof} We construct a gap-preserving reduction from \textsc{Dominating Set} restricted to subcubic graphs. Let $G$ be an instance of this problem. Clearly, we may assume that $G$ is a subcubic graph of order $n$ without isolated vertices. We build a graph $G'$ by replacing each vertex $v$ of $G$ with the gadget $G_{v}$ depicted in \Cref{fig:approx} and by identifying edges incident to different gadgets and corresponding to the same $e_{i} \in E(G)$ (vertices of degree $1$ or $2$ are treated similarly). This means that the contraction of all the edges in the gadgets in $G'$ results in the graph $G$. Note that $V(G_{v}) = V_{1} \cup V_{2} \cup V_{3}$, where each $V_{i} = \{v_{i}, u_{i}, w_{i}, a_{i}, b_{i}, c_{i}\}$ naturally corresponds to the edge $e_{i}$ incident to $v$. Observe also that $G'$ is subcubic and is in fact the line graph of the graph $G''$ obtained first by a $1$-subdivision of $G$ and then by replacing each vertex $v$ of $G$ with the gadget depicted in \Cref{fig:line} and identifying the $d(v)$ edges incident to $v$ with $d(v)$ edges incident to the gadget.   
 
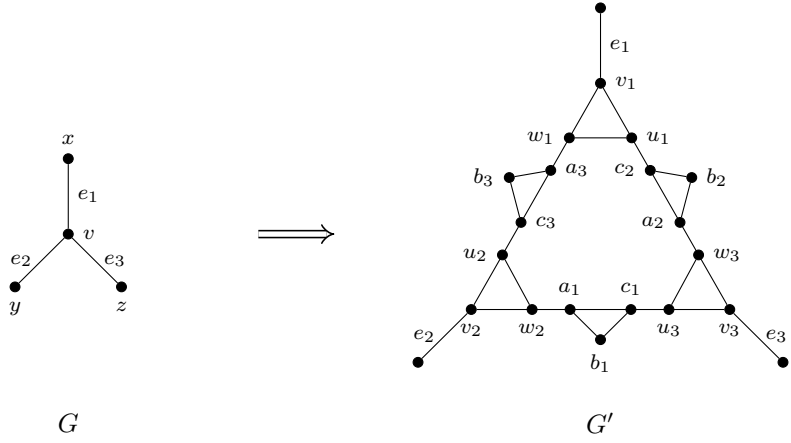
\begin{figure}[h!]
\centering
\begin{tikzpicture}
\node[circ,label=right:{\footnotesize $v$}] (v) at (-1,0) {};
\node[circ,label=below:{\footnotesize $y$}] (y) at (-1.7,-.7) {};
\node[circ,label=below:{\footnotesize $z$}] (z) at (-.3,-.7) {};
\node[circ,label=above:{\footnotesize $x$}] (x) at (-1,1) {};
\draw (v) -- (x) node[midway,right] {\footnotesize $e_1$};
\draw (v) -- (y) node[midway,left] {\footnotesize $e_2$};
\draw (v) -- (z) node[midway,right] {\footnotesize $e_3$};

\node[draw=none] at (-1,-2.5) {$G$};

\draw[-Implies,line width=.6pt,double distance=2pt] (1.5,0) -- (2.5,0);

\node[circ,label=right:{\footnotesize $v_1$}] (v1) at (6,2) {};
\node[circ,label=below:{\footnotesize $v_2$}] (v2) at (4.3,-1) {};
\node[circ,label=below:{\footnotesize $v_3$}] (v3) at (7.7,-1) {}; 
\node[circ,label=below:{\footnotesize $w_2$}] (w2) at (5.1,-1) {};
\node[circ,label=above:{\footnotesize $a_1$}] (a1) at (5.6,-1) {};
\node[circ,label=above:{\footnotesize $c_1$}] (c1) at (6.4,-1) {};
\node[circ,label=below:{\footnotesize $u_3$}] (u3) at (6.9,-1) {};
\node[circ,label=below:{\footnotesize $b_1$}] (b1) at (6,-1.4) {};
\draw (v2) -- (v3)
(a1) -- (b1)
(b1) -- (c1);

\draw (v2) -- (v1) node[circ,label=left:{\footnotesize $u_2$},pos=.23] (u2) {}
node[circ,label=right:{\footnotesize $c_3$},pos=.38] (c3) {}
node[circ,label=right:{\footnotesize $a_3$},pos=.62] (a3) {}
node[circ,label=left:{\footnotesize $w_1$},pos=.77] (w1) {};

\draw (v3) -- (v1) node[circ,label=right:{\footnotesize $w_3$},pos=.23] (w3) {}
node[circ,label=left:{\footnotesize $a_2$},pos=.38] (a2) {}
node[circ,label=left:{\footnotesize $c_2$},pos=.62] (c2) {}
node[circ,label=right:{\footnotesize $u_1$},pos=.77] (u1) {};

\node[circ,label=right:{\footnotesize $b_2$}] (b2) at (7.2,.75) {};
\draw (a2) -- (b2)
(b2) -- (c2);

\node[circ,label=left:{\footnotesize $b_3$}] (b3) at (4.8,.75) {};
\draw (a3) -- (b3)
(b3) -- (c3);

\draw (u1) -- (w1) 
(u2) -- (w2)
(u3) -- (w3);

\node[circ] (e1) at (6,3) {};
\draw (v1) -- (e1) node[midway,right] {\footnotesize $e_1$};
\node[circ] (e2) at (3.6,-1.7) {};
\draw (v2) -- (e2) node[midway,left] {\footnotesize $e_2$};
\node[circ] (e3) at (8.4,-1.7) {};
\draw (v3) -- (e3) node[midway,right] {\footnotesize $e_3$};

\node[draw=none] at (6,-2.5) {$G'$};
\end{tikzpicture}
\caption{The gadget $G_{v}$ in $G'$ replacing the cubic vertex $v$ of $G$.}
\label{fig:approx}
\end{figure}

\begin{center}
\begin{figure}[h]
\centering
\begin{tikzpicture}
\node[circ] (v1) at (6,2) {};
\node[circ] (v2) at (4.3,-1) {};
\node[circ] (v3) at (7.7,-1) {};

\draw (v1) -- (v2) node[pos=.25,circ] {}
node[pos=.5,circ] (a1) {}
node[pos=.75,circ] {};

\draw (v1) -- (v3) node[pos=.25,circ] {}
node[pos=.5,circ] (a2) {}
node[pos=.75,circ] {};

\draw (v2) -- (v3) node[pos=.25,circ] {}
node[pos=.5,circ] (a3) {}
node[pos=.75,circ] {};

\node[circ] (e1) at (6,3) {};
\draw (v1) -- (e1); 
\node[circ] (e2) at (3.6,-1.7) {};
\draw (v2) -- (e2);
\node[circ] (e3) at (8.4,-1.7) {};
\draw (v3) -- (e3);

\node[circ] (b1) at (6,-1.4) {};
\draw (b1) -- (a3);

\node[circ] (b2) at (7.2,.75) {};
\draw (b2) -- (a2);

\node[circ] (b3) at (4.8,.75) {};
\draw (b3) -- (a1);

\end{tikzpicture}
\caption{The gadget in $G''$ replacing each vertex $v$ of $G$.} 
\label{fig:line}
\end{figure}
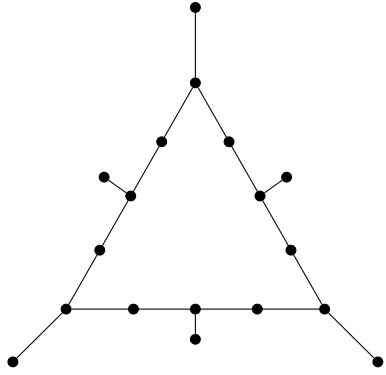
\end{center}

We claim that $\gamma_{t2}(G') = \gamma(G) + 5n$. Suppose first that $D$ is a minimum dominating set of $G$. We construct a semi-TD-set $D'$ of $G'$ as follows. For each vertex $v \in D$, we select the six vertices of $G_{v}$ depicted in light gray in \Cref{fig:dom}, i.e. we add $D_{v} = \{v_{1}, v_{2}, v_{3}, c_{1}, c_{2}, c_{3}\}$ to $D'$. Observe that each vertex of $G_{v}$ is dominated by one vertex in $D_{v}$ and each vertex in $D_{v}$ is within distance two of another vertex of $D_v$. For each vertex $v \notin D$, there exists a vertex $x \in D$ adjacent to $v$, say $e_{1} = xv$. In the previous step, we have selected all the vertices $x_{i}$ of $G_{x}$ corresponding to $x$ and so, by further selecting the five vertices of $G_{v}$ depicted in light gray in \Cref{fig:notdom}, we obtain a subset dominating each vertex of $G_{v}$. Moreover, each newly selected vertex is within distance two of another newly selected vertex. This implies that $D'$ is a semi-TD-set of size $\gamma(G) + 5n$ and so $\gamma_{t2}(G') \leq \gamma(G) + 5n$. 

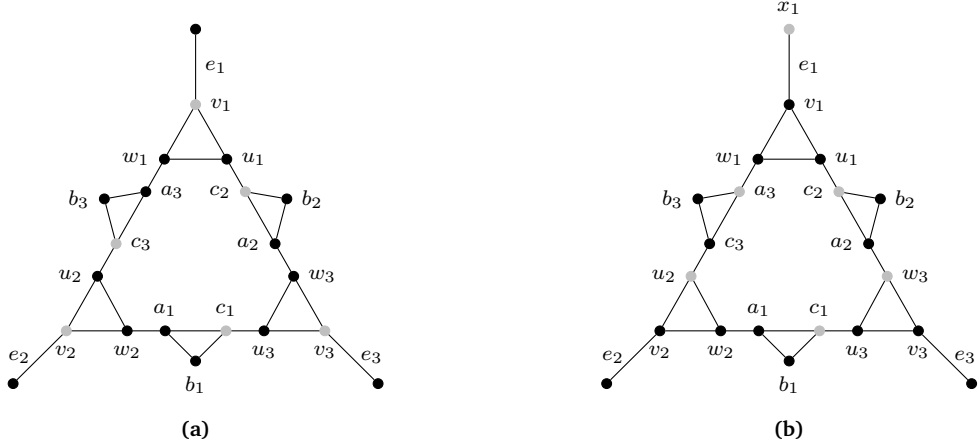
\begin{figure}[h!]
\centering
\hspace*{\fill}
\begin{subfigure}[b]{0.45\textwidth}
\centering
\begin{tikzpicture}
\node[circg,label=right:{\footnotesize $v_1$}] (v1) at (6,2) {};
\node[circg,label=below:{\footnotesize $v_2$}] (v2) at (4.3,-1) {};
\node[circg,label=below:{\footnotesize $v_3$}] (v3) at (7.7,-1) {}; 
\node[circ,label=below:{\footnotesize $w_2$}] (w2) at (5.1,-1) {};
\node[circ,label=above:{\footnotesize $a_1$}] (a1) at (5.6,-1) {};
\node[circg,label=above:{\footnotesize $c_1$}] (c1) at (6.4,-1) {};
\node[circ,label=below:{\footnotesize $u_3$}] (u3) at (6.9,-1) {};
\node[circ,label=below:{\footnotesize $b_1$}] (b1) at (6,-1.4) {};
\draw (v2) -- (c1)
(c1) -- (v3)
(a1) -- (b1)
(b1) -- (c1);

\draw (v2) -- (v1) node[circ,label=left:{\footnotesize $u_2$},pos=.23] (u2) {}
node[circg,label=right:{\footnotesize $c_3$},pos=.38] (c3) {}
node[circ,label=right:{\footnotesize $a_3$},pos=.62] (a3) {}
node[circ,label=left:{\footnotesize $w_1$},pos=.77] (w1) {};

\draw (v3) -- (v1) node[circ,label=right:{\footnotesize $w_3$},pos=.23] (w3) {}
node[circ,label=left:{\footnotesize $a_2$},pos=.38] (a2) {}
node[circg,label=left:{\footnotesize $c_2$},pos=.62] (c2) {}
node[circ,label=right:{\footnotesize $u_1$},pos=.77] (u1) {};

\node[circ,label=right:{\footnotesize $b_2$}] (b2) at (7.2,.75) {};
\draw (a2) -- (b2)
(b2) -- (c2);

\node[circ,label=left:{\footnotesize $b_3$}] (b3) at (4.8,.75) {};
\draw (a3) -- (b3)
(b3) -- (c3);

\draw (u1) -- (w1) 
(u2) -- (w2)
(u3) -- (w3);

\node[circ] (e1) at (6,3) {};
\draw (v1) -- (e1) node[midway,right] {\footnotesize $e_1$};
\node[circ] (e2) at (3.6,-1.7) {};
\draw (v2) -- (e2) node[midway,left] {\footnotesize $e_2$};
\node[circ] (e3) at (8.4,-1.7) {};
\draw (v3) -- (e3) node[midway,right] {\footnotesize $e_3$};
\end{tikzpicture}
\caption{}
\label{fig:dom}
\end{subfigure}
\hfill
\begin{subfigure}[b]{0.45\textwidth}
\centering
\begin{tikzpicture}
\node[circ,label=right:{\footnotesize $v_1$},] (v1) at (6,2) {};
\node[circ,label=below:{\footnotesize $v_2$}] (v2) at (4.3,-1) {};
\node[circ,label=below:{\footnotesize $v_3$}] (v3) at (7.7,-1) {}; 
\node[circ,label=below:{\footnotesize $w_2$}] (w2) at (5.1,-1) {};
\node[circ,label=above:{\footnotesize $a_1$}] (a1) at (5.6,-1) {};
\node[circg,label=above:{\footnotesize $c_1$}] (c1) at (6.4,-1) {};
\node[circ,label=below:{\footnotesize $u_3$}] (u3) at (6.9,-1) {};
\node[circ,label=below:{\footnotesize $b_1$}] (b1) at (6,-1.4) {};
\draw (v2) -- (c1)
(c1) -- (v3)
(a1) -- (b1)
(b1) -- (c1);

\draw (v2) -- (v1) node[circg,label=left:{\footnotesize $u_2$},pos=.23] (u2) {}
node[circ,label=right:{\footnotesize $c_3$},pos=.38] (c3) {}
node[circg,label=right:{\footnotesize $a_3$},pos=.62] (a3) {}
node[circ,label=left:{\footnotesize $w_1$},pos=.77] (w1) {};

\draw (v3) -- (v1) node[circg,label=right:{\footnotesize $w_3$},pos=.23] (w3) {}
node[circ,label=left:{\footnotesize $a_2$},pos=.38] (a2) {}
node[circg,label=left:{\footnotesize $c_2$},pos=.62] (c2) {}
node[circ,label=right:{\footnotesize $u_1$},pos=.77] (u1) {};

\node[circ,label=right:{\footnotesize $b_2$}] (b2) at (7.2,.75) {};
\draw (a2) -- (b2)
(b2) -- (c2);

\node[circ,label=left:{\footnotesize $b_3$}] (b3) at (4.8,.75) {};
\draw (a3) -- (b3)
(b3) -- (c3);

\draw (u1) -- (w1) 
(u2) -- (w2)
(u3) -- (w3);

\node[circg,label=above:{\footnotesize $x_1$}] (e1) at (6,3) {};
\draw (v1) -- (e1) node[midway,right] {\footnotesize $e_1$};
\node[circ] (e2) at (3.6,-1.7) {};
\draw (v2) -- (e2) node[midway,left] {\footnotesize $e_2$};
\node[circ] (e3) at (8.4,-1.7) {};
\draw (v3) -- (e3) node[midway,right] {\footnotesize $e_3$};
\end{tikzpicture}
\caption{}
\label{fig:notdom}
\end{subfigure}
\hspace*{\fill}
\caption{Construction of a semi-TD-set of $G'$ from a dominating set of $G$.}
\end{figure} 

Conversely, let $D'$ be a minimum semi-TD-set of $G'$. We first observe the following.

\begin{observation}
\label{obs:si}
For each $i \in \{1,2,3\}$, at least one vertex of $S_i = \{a_i,b_i,c_i\}$ belongs to $D'$.
\end{observation}

We next prove the following claim:

\begin{claim}\label{claimcount}
$|D' \cap V(G_{v})| \geq 5$, for each $v \in V(G)$. Moreover, equality holds only if $D' \cap \{v_{1}, v_{2}, v_{3}\} = \varnothing$. 
\end{claim}

By Observation \ref{obs:si}, $|D' \cap S_i| \geq 1$ for each $i \in \{1,2,3\}$. Since for any $j \neq i$, $d_{G'}(S_{i}, S_{j}) \geq 3$, it is easy to see that $D'$ must further contain at least two vertices of $V(G_{v})$. Therefore, we have $|D' \cap V(G_{v})| \geq 5$. 

We now show that if $D'$ contains a vertex in $\{v_{1}, v_{2}, v_{3}\}$, then $|D' \cap V(G_{v})| \geq 6$. Suppose first that $D' \cap \{v_i,w_i,u_i\} \neq \varnothing$, for each $i \in \{1,2,3\}$. Observation \ref{obs:si} immediately implies that $|D' \cap V(G_{v})| \geq 3 + 3 = 6$. Suppose finally that there exists $i \in \{1,2,3\}$ such that $D' \cap \{v_i,w_i,u_i\} = \varnothing$. We then have that $\{a_{i-1}, c_{i+1}\} \subseteq D'$ (indices are taken modulo $3$), as $w_i$ and $u_i$ are dominated. But then the witness of $a_{i-1}$ and that of $c_{i+1}$ are distinct (as $D' \cap \{v_i,w_i,u_i\} = \varnothing$) and different from $v_{i-1}$, $v_{i+1}$ and any vertex of $S_i$. Since by Observation~\ref{obs:si} at least one vertex in $S_i$ belongs to $D'$ and, by assumption, $|D' \cap \{v_{1}, v_{2}, v_{3}\}| \geq 1$, it then follows that $|D' \cap V(G_{v})| \geq 2 + 2 + 1 + 1 = 6$. $\lozenge$

\smallskip

Denoting by $t$ the number of gadgets $G_{v}$ of $G'$ such that $|D' \cap V(G_{v})| \geq 6$, \Cref{claimcount} implies that $|D'| \geq 6t + 5(n - t) = t + 5n$. Let now $D$ be the set of vertices $v$ of $G$ corresponding to those gadgets $G_{v}$ in $G'$ such that $|D' \cap V(G_{v})| \geq 6$. We claim that $D$ is a dominating set of $G$. Indeed, if $v \notin D$, then $|D' \cap V(G_{v})| = 5$ and \Cref{claimcount} implies that $D' \cap \{v_{1}, v_{2}, v_{3}\} = \varnothing$. Therefore, at least one vertex of $\{v_{1}, v_{2}, v_{3}\}$ is dominated by a vertex in a different gadget and since this vertex is of the form $x_{i}$, for some $x$ adjacent to $v$, we have that $|D' \cap V(G_{x})| \geq 6$ and so $x \in D$. Finally, it is enough to notice that $\gamma(G) \leq |D| = t \leq |D'| - 5n = \gamma_{t2}(G') - 5n$ which proves our initial claim.      

For a given subcubic graph $G$ of order $n$, \citet{CC08} showed that it is $\mathsf{NP}$-hard to decide whether $\gamma(G) > 0.28792798n$ or $\gamma(G) < 0.28719008n$ and so it is $\mathsf{NP}$-hard to decide whether $\gamma_{t2}(G') > 5.28792798n$ or $\gamma_{t2}(G') < 5.28719008n$. Therefore, by \Cref{gapred}, there is no $1.00013956$-approximation algorithm for \textsc{Semitotal Dominating Set}, unless $\mathsf{P} = \mathsf{NP}$, even when restricted to subcubic line graphs of bipartite graphs.
\end{proof}

By slightly modifying the gadget used in the previous proof, we now show that \textsc{Semitotal Dominating Set} is $\mathsf{APX}$-complete even when restricted to cubic graphs.

\setcounter{claim}{0}

\begin{theorem}\label{semicubic} \textsc{Semitotal Dominating Set} is $\mathsf{APX}$-complete when restricted to cubic graphs.
\end{theorem}

\begin{proof} \citet{HP17} showed that \textsc{Semitotal Dominating Set} restricted to graphs with bounded degree belongs to $\mathsf{APX}$. It is therefore enough to construct an L-reduction from \textsc{Dominating Set} restricted to cubic graphs, which is known to be $\mathsf{APX}$-complete \cite{AK00}. Let $G$ be an instance of this problem with $n$ vertices. We build a graph $G'$ by replacing each vertex $v$ of $G$ with the gadget $G_{v}$ depicted in \Cref{fig:gadgetcubic} and by identifying edges incident to different gadgets and corresponding to the same $e_{i} \in E(G)$. This means that the contraction of all the edges in the gadgets in $G'$ results in the graph $G$. 

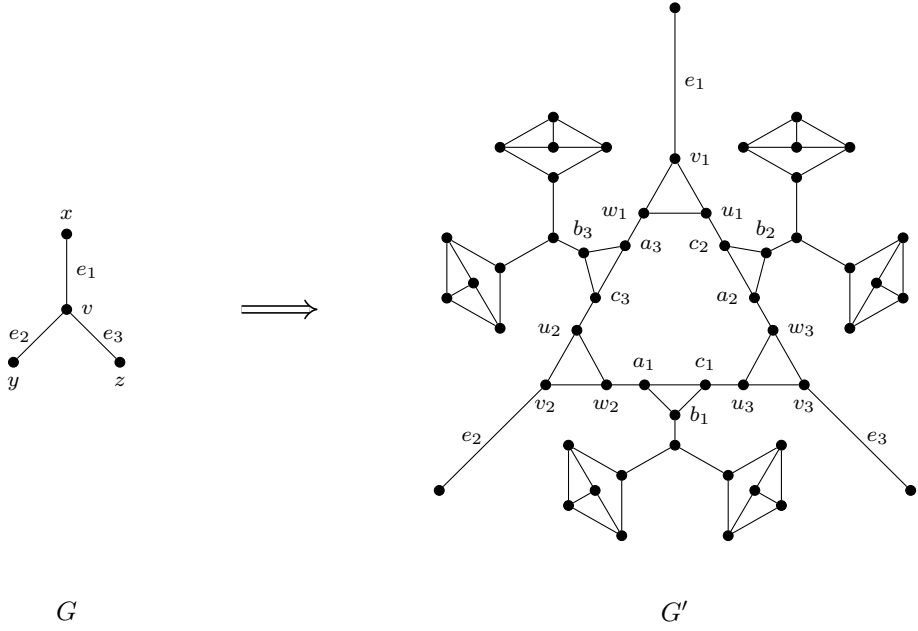
\begin{figure}[h!]
\centering
\begin{tikzpicture}
\node[circ,label=right:{\footnotesize $v$}] (v) at (-2,0) {};
\node[circ,label=below:{\footnotesize $y$}] (y) at (-2.7,-.7) {};
\node[circ,label=below:{\footnotesize $z$}] (z) at (-1.3,-.7) {};
\node[circ,label=above:{\footnotesize $x$}] (x) at (-2,1) {};
\draw (v) -- (x) node[midway,right] {\footnotesize $e_1$};
\draw (v) -- (y) node[midway,left] {\footnotesize $e_2$};
\draw (v) -- (z) node[midway,right] {\footnotesize $e_3$};

\node[draw=none] at (-2,-4) {$G$};

\draw[-Implies,line width=.6pt,double distance=2pt] (.3,0) -- (1.3,0);

\node[circ,label=right:{\footnotesize $v_1$}] (v1) at (6,2) {};
\node[circ,label=below:{\footnotesize $v_2$}] (v2) at (4.3,-1) {};
\node[circ,label=below:{\footnotesize $v_3$}] (v3) at (7.7,-1) {}; 
\node[circ,label=below:{\footnotesize $w_2$}] (w2) at (5.1,-1) {};
\node[circ,label=above:{\footnotesize $a_1$}] (a1) at (5.6,-1) {};
\node[circ,label=above:{\footnotesize $c_1$}] (c1) at (6.4,-1) {};
\node[circ,label=below:{\footnotesize $u_3$}] (u3) at (6.9,-1) {};
\node[circ,label=right:{\footnotesize $b_1$}] (b1) at (6,-1.4) {};
\draw (v2) -- (v3)
(a1) -- (b1)
(b1) -- (c1);

\draw (v2) -- (v1) node[circ,label=left:{\footnotesize $u_2$},pos=.23] (u2) {}
node[circ,label=right:{\footnotesize $c_3$},pos=.38] (c3) {}
node[circ,label=right:{\footnotesize $a_3$},pos=.62] (a3) {}
node[circ,label=left:{\footnotesize $w_1$},pos=.77] (w1) {};

\draw (v3) -- (v1) node[circ,label=right:{\footnotesize $w_3$},pos=.23] (w3) {}
node[circ,label=left:{\footnotesize $a_2$},pos=.38] (a2) {}
node[circ,label=left:{\footnotesize $c_2$},pos=.62] (c2) {}
node[circ,label=right:{\footnotesize $u_1$},pos=.77] (u1) {};

\node[circ,label=above:{\footnotesize $b_2$}] (b2) at (7.2,.75) {};
\draw (a2) -- (b2)
(b2) -- (c2);

\node[circ,label=above:{\footnotesize $b_3$}] (b3) at (4.8,.75) {};
\draw (a3) -- (b3)
(b3) -- (c3);

\draw (u1) -- (w1) 
(u2) -- (w2)
(u3) -- (w3);

\node[circ] (e1) at (6,4) {};
\draw (v1) -- (e1) node[midway,right] {\footnotesize $e_1$};
\node[circ] (e2) at (2.9,-2.4) {};
\draw (v2) -- (e2) node[midway,left] {\footnotesize $e_2$};
\node[circ] (e3) at (9.1,-2.4) {};
\draw (v3) -- (e3) node[midway,right] {\footnotesize $e_3$};

\node[circ] (x2) at (7.6,.95) {};
\draw (b2) -- (x2);

\node[circ] (x3) at (7.6,1.75) {};
\node[circ] (x5) at (6.9,2.15) {};
\node[circ] (x6) at (8.3,2.15) {};
\node[circ] (x7) at (7.6,2.55) {};
\draw (x2) -- (x3)
(x3) -- (x5)
(x3) -- (x6)
(x5) -- (x7)
(x6) -- (x7)
(x5) -- (x6) node[circ,midway] (x8) {};
\draw (x8) -- (x7);

\node[circ] (x4) at (8.3,.55) {};
\node[circ] (x9) at (8.3,-.25) {};
\node[circ] (x10) at (9,.95) {};
\node[circ] (x11) at (9,.15) {};
\draw (x2) -- (x4)
(x4) -- (x9)
(x4) -- (x10)
(x9) -- (x11)
(x10) -- (x11)
(x9) -- (x10) node[circ,midway] (x12) {};
\draw (x12) -- (x11);

\node[circ] (y2) at (6,-1.8) {};
\draw (b1) -- (y2);

\node[circ] (y3) at (6.7,-2.2) {};
\node[circ] (y5) at (6.7,-3) {};
\node[circ] (y6) at (7.4,-1.8) {};
\node[circ] (y7) at (7.4,-2.6) {};
\draw (y2) -- (y3)
(y3) -- (y5) 
(y3) -- (y6)
(y5) -- (y7)
(y6) -- (y7)
(y5) -- (y6) node[circ,midway] (y8) {};
\draw (y8) -- (y7);

\node[circ] (y4) at (5.3,-2.2) {};
\node[circ] (y9) at (5.3,-3) {};
\node[circ] (y10) at (4.6,-1.8) {};
\node[circ] (y11) at (4.6,-2.6) {};
\draw (y2) -- (y4)
(y4) -- (y9)
(y4) -- (y10)
(y9) -- (y11)
(y10) -- (y11)
(y9) -- (y10) node[circ,midway] (y12) {};
\draw (y12) -- (y11);

\node[circ] (z2) at (4.4,.95) {};
\draw (b3) -- (z2);

\node[circ] (z3) at (4.4,1.75) {};
\node[circ] (z5) at (5.1,2.15) {};
\node[circ] (z6) at (3.7,2.15) {};
\node[circ] (z7) at (4.4,2.55) {};
\draw (z2) -- (z3)
(z3) -- (z5)
(z3) -- (z6)
(z5) -- (z7)
(z6) -- (z7)
(z5) -- (z6) node[circ,midway] (z8) {};
\draw (z8) -- (z7);

\node[circ] (z4) at (3.7,.55) {};
\node[circ] (z9) at (3.7,-.25) {};
\node[circ] (z10) at (3,.95) {};
\node[circ] (z11) at (3,.15) {};
\draw (z2) -- (z4)
(z4) -- (z9)
(z4) -- (z10)
(z9) -- (z11)
(z10) -- (z11)
(z9) -- (z10) node[circ,midway] (z12) {};
\draw (z12) -- (z11);

\node[draw=none] at (6,-4) {$G'$};
\end{tikzpicture}
\caption{The gadget $G_{v}$ in $G'$ replacing the vertex $v$ of $G$.}
\label{fig:gadgetcubic}
\end{figure}

We claim that $\gamma_{t2}(G') = \gamma(G) + 17n$. Suppose first that $D$ is a minimum dominating set $D$ of $G$. We construct a semi-TD-set $D'$ of $G'$ as follows. For each vertex $v \in D$, we select the set $D_{v} \subset V(G_{v})$ consisting of the vertices depicted in light gray in \Cref{fig:cubicdom}. Observe that each vertex of $G_{v}$ is dominated by one vertex in $D_{v}$ and each vertex in $D_{v}$ has a witness in $D_{v}$. For each vertex $v \notin D$, there exists a vertex $x \in D$ adjacent to $v$, say $e_{1} = xv$. In the previous step, we have selected all the vertices $x_{i}$ of $G_{x}$ corresponding to $x$ and so, by further selecting the vertices of $G_{v}$ depicted in light gray in \Cref{fig:cubicnotdom}, we obtain a subset dominating each vertex of $G_{v}$. Moreover, each newly selected vertex is within distance two of another newly selected vertex. This implies that $D'$ is a semi-TD-set of size $\gamma(G) + 17n$ and so $\gamma_{t2}(G') \leq \gamma(G) + 17n$. \\

\begin{figure}[h!]
\centering
\hspace*{\fill}
\begin{subfigure}[b]{0.45\textwidth}
\centering
\begin{tikzpicture}
\node[circg,label=right:{\footnotesize $v_1$}] (v1) at (6,2) {};
\node[circg,label=below:{\footnotesize $v_2$}] (v2) at (4.3,-1) {};
\node[circg,label=below:{\footnotesize $v_3$}] (v3) at (7.7,-1) {}; 
\node[circ,label=below:{\footnotesize $w_2$}] (w2) at (5.1,-1) {};
\node[circ,label=above:{\footnotesize $a_1$}] (a1) at (5.6,-1) {};
\node[circg,label=above:{\footnotesize $c_1$}] (c1) at (6.4,-1) {};
\node[circ,label=below:{\footnotesize $u_3$}] (u3) at (6.9,-1) {};
\node[circ,label=right:{\footnotesize $b_1$}] (b1) at (6,-1.4) {};
\draw (v2) -- (c1)
(c1) -- (v3)
(a1) -- (b1)
(b1) -- (c1);

\draw (v2) -- (v1) node[circ,label=left:{\footnotesize $u_2$},pos=.23] (u2) {}
node[circg,label=right:{\footnotesize $c_3$},pos=.38] (c3) {}
node[circ,label=right:{\footnotesize $a_3$},pos=.62] (a3) {}
node[circ,label=left:{\footnotesize $w_1$},pos=.77] (w1) {};

\draw (v3) -- (v1) node[circ,label=right:{\footnotesize $w_3$},pos=.23] (w3) {}
node[circ,label=left:{\footnotesize $a_2$},pos=.38] (a2) {}
node[circg,label=left:{\footnotesize $c_2$},pos=.62] (c2) {}
node[circ,label=right:{\footnotesize $u_1$},pos=.77] (u1) {};

\node[circ,label=above:{\footnotesize $b_2$}] (b2) at (7.2,.75) {};
\draw (a2) -- (b2)
(b2) -- (c2);

\node[circ,label=above:{\footnotesize $b_3$}] (b3) at (4.8,.75) {};
\draw (a3) -- (b3)
(b3) -- (c3);

\draw (u1) -- (w1) 
(u2) -- (w2)
(u3) -- (w3);

\node[circ] (e1) at (6,4) {};
\draw (v1) -- (e1) node[midway,right] {\footnotesize $e_1$};
\node[circ] (e2) at (2.9,-2.4) {};
\draw (v2) -- (e2) node[midway,left] {\footnotesize $e_2$};
\node[circ] (e3) at (9.1,-2.4) {};
\draw (v3) -- (e3) node[midway,right] {\footnotesize $e_3$};

\node[circ] (x2) at (7.6,.95) {};
\draw (b2) -- (x2);

\node[circg] (x3) at (7.6,1.75) {};
\node[circ] (x5) at (6.9,2.15) {};
\node[circ] (x6) at (8.3,2.15) {};
\node[circ] (x7) at (7.6,2.55) {};
\draw (x2) -- (x3)
(x3) -- (x5)
(x3) -- (x6)
(x5) -- (x7)
(x6) -- (x7)
(x5) -- (x6) node[circg,midway] (x8) {};
\draw (x8) -- (x7);

\node[circg] (x4) at (8.3,.55) {};
\node[circ] (x9) at (8.3,-.25) {};
\node[circ] (x10) at (9,.95) {};
\node[circ] (x11) at (9,.15) {};
\draw (x2) -- (x4)
(x4) -- (x9)
(x4) -- (x10)
(x9) -- (x11)
(x10) -- (x11)
(x9) -- (x10) node[circg,midway] (x12) {};
\draw (x12) -- (x11);

\node[circ] (y2) at (6,-1.8) {};
\draw (b1) -- (y2);

\node[circg] (y3) at (6.7,-2.2) {};
\node[circ] (y5) at (6.7,-3) {};
\node[circ] (y6) at (7.4,-1.8) {};
\node[circ] (y7) at (7.4,-2.6) {};
\draw (y2) -- (y3)
(y3) -- (y5) 
(y3) -- (y6)
(y5) -- (y7)
(y6) -- (y7)
(y5) -- (y6) node[circg,midway] (y8) {};
\draw (y8) -- (y7);

\node[circg] (y4) at (5.3,-2.2) {};
\node[circ] (y9) at (5.3,-3) {};
\node[circ] (y10) at (4.6,-1.8) {};
\node[circ] (y11) at (4.6,-2.6) {};
\draw (y2) -- (y4)
(y4) -- (y9)
(y4) -- (y10)
(y9) -- (y11)
(y10) -- (y11)
(y9) -- (y10) node[circg,midway] (y12) {};
\draw (y12) -- (y11);

\node[circ] (z2) at (4.4,.95) {};
\draw (b3) -- (z2);

\node[circg] (z3) at (4.4,1.75) {};
\node[circ] (z5) at (5.1,2.15) {};
\node[circ] (z6) at (3.7,2.15) {};
\node[circ] (z7) at (4.4,2.55) {};
\draw (z2) -- (z3)
(z3) -- (z5)
(z3) -- (z6)
(z5) -- (z7)
(z6) -- (z7)
(z5) -- (z6) node[circg,midway] (z8) {};
\draw (z8) -- (z7);

\node[circg] (z4) at (3.7,.55) {};
\node[circ] (z9) at (3.7,-.25) {};
\node[circ] (z10) at (3,.95) {};
\node[circ] (z11) at (3,.15) {};
\draw (z2) -- (z4)
(z4) -- (z9)
(z4) -- (z10)
(z9) -- (z11)
(z10) -- (z11)
(z9) -- (z10) node[circg,midway] (z12) {};
\draw (z12) -- (z11);
\end{tikzpicture}
\caption{}
\label{fig:cubicdom}
\end{subfigure}
\hfill
\begin{subfigure}[b]{0.45\textwidth}
\centering
\begin{tikzpicture}
\node[circ,label=right:{\footnotesize $v_1$}] (v1) at (6,2) {};
\node[circ,label=below:{\footnotesize $v_2$}] (v2) at (4.3,-1) {};
\node[circ,label=below:{\footnotesize $v_3$}] (v3) at (7.7,-1) {}; 
\node[circ,label=below:{\footnotesize $w_2$}] (w2) at (5.1,-1) {};
\node[circ,label=above:{\footnotesize $a_1$}] (a1) at (5.6,-1) {};
\node[circg,label=above:{\footnotesize $c_1$}] (c1) at (6.4,-1) {};
\node[circ,label=below:{\footnotesize $u_3$}] (u3) at (6.9,-1) {};
\node[circ,label=right:{\footnotesize $b_1$}] (b1) at (6,-1.4) {};
\draw (v2) -- (c1)
(c1) -- (v3)
(a1) -- (b1)
(b1) -- (c1);

\draw (v2) -- (v1) node[circg,label=left:{\footnotesize $u_2$},pos=.23] (u2) {}
node[circ,label=right:{\footnotesize $c_3$},pos=.38] (c3) {}
node[circg,label=right:{\footnotesize $a_3$},pos=.62] (a3) {}
node[circ,label=left:{\footnotesize $w_1$},pos=.77] (w1) {};

\draw (v3) -- (v1) node[circg,label=right:{\footnotesize $w_3$},pos=.23] (w3) {}
node[circ,label=left:{\footnotesize $a_2$},pos=.38] (a2) {}
node[circg,label=left:{\footnotesize $c_2$},pos=.62] (c2) {}
node[circ,label=right:{\footnotesize $u_1$},pos=.77] (u1) {};

\node[circ,label=above:{\footnotesize $b_2$}] (b2) at (7.2,.75) {};
\draw (a2) -- (b2)
(b2) -- (c2);

\node[circ,label=above:{\footnotesize $b_3$}] (b3) at (4.8,.75) {};
\draw (a3) -- (b3)
(b3) -- (c3);

\draw (u1) -- (w1) 
(u2) -- (w2)
(u3) -- (w3);

\node[circg] (e1) at (6,4) {};
\draw (v1) -- (e1) node[midway,right] {\footnotesize $e_1$};
\node[circ] (e2) at (2.9,-2.4) {};
\draw (v2) -- (e2) node[midway,left] {\footnotesize $e_2$};
\node[circ] (e3) at (9.1,-2.4) {};
\draw (v3) -- (e3) node[midway,right] {\footnotesize $e_3$};

\node[circ] (x2) at (7.6,.95) {};
\draw (b2) -- (x2);

\node[circg] (x3) at (7.6,1.75) {};
\node[circ] (x5) at (6.9,2.15) {};
\node[circ] (x6) at (8.3,2.15) {};
\node[circ] (x7) at (7.6,2.55) {};
\draw (x2) -- (x3)
(x3) -- (x5)
(x3) -- (x6)
(x5) -- (x7)
(x6) -- (x7)
(x5) -- (x6) node[circg,midway] (x8) {};
\draw (x8) -- (x7);

\node[circg] (x4) at (8.3,.55) {};
\node[circ] (x9) at (8.3,-.25) {};
\node[circ] (x10) at (9,.95) {};
\node[circ] (x11) at (9,.15) {};
\draw (x2) -- (x4)
(x4) -- (x9)
(x4) -- (x10)
(x9) -- (x11)
(x10) -- (x11)
(x9) -- (x10) node[circg,midway] (x12) {};
\draw (x12) -- (x11);

\node[circ] (y2) at (6,-1.8) {};
\draw (b1) -- (y2);

\node[circg] (y3) at (6.7,-2.2) {};
\node[circ] (y5) at (6.7,-3) {};
\node[circ] (y6) at (7.4,-1.8) {};
\node[circ] (y7) at (7.4,-2.6) {};
\draw (y2) -- (y3)
(y3) -- (y5) 
(y3) -- (y6)
(y5) -- (y7)
(y6) -- (y7)
(y5) -- (y6) node[circg,midway] (y8) {};
\draw (y8) -- (y7);

\node[circg] (y4) at (5.3,-2.2) {};
\node[circ] (y9) at (5.3,-3) {};
\node[circ] (y10) at (4.6,-1.8) {};
\node[circ] (y11) at (4.6,-2.6) {};
\draw (y2) -- (y4)
(y4) -- (y9)
(y4) -- (y10)
(y9) -- (y11)
(y10) -- (y11)
(y9) -- (y10) node[circg,midway] (y12) {};
\draw (y12) -- (y11);

\node[circ] (z2) at (4.4,.95) {};
\draw (b3) -- (z2);

\node[circg] (z3) at (4.4,1.75) {};
\node[circ] (z5) at (5.1,2.15) {};
\node[circ] (z6) at (3.7,2.15) {};
\node[circ] (z7) at (4.4,2.55) {};
\draw (z2) -- (z3)
(z3) -- (z5)
(z3) -- (z6)
(z5) -- (z7)
(z6) -- (z7)
(z5) -- (z6) node[circg,midway] (z8) {};
\draw (z8) -- (z7);

\node[circg] (z4) at (3.7,.55) {};
\node[circ] (z9) at (3.7,-.25) {};
\node[circ] (z10) at (3,.95) {};
\node[circ] (z11) at (3,.15) {};
\draw (z2) -- (z4)
(z4) -- (z9)
(z4) -- (z10)
(z9) -- (z11)
(z10) -- (z11)
(z9) -- (z10) node[circg,midway] (z12) {};
\draw (z12) -- (z11);
\end{tikzpicture}
\caption{}
\label{fig:cubicnotdom}
\end{subfigure}
\hspace*{\fill}
\caption{Construction of a semi-TD-set of $G'$ from a dominating set of $G$.}
\end{figure}
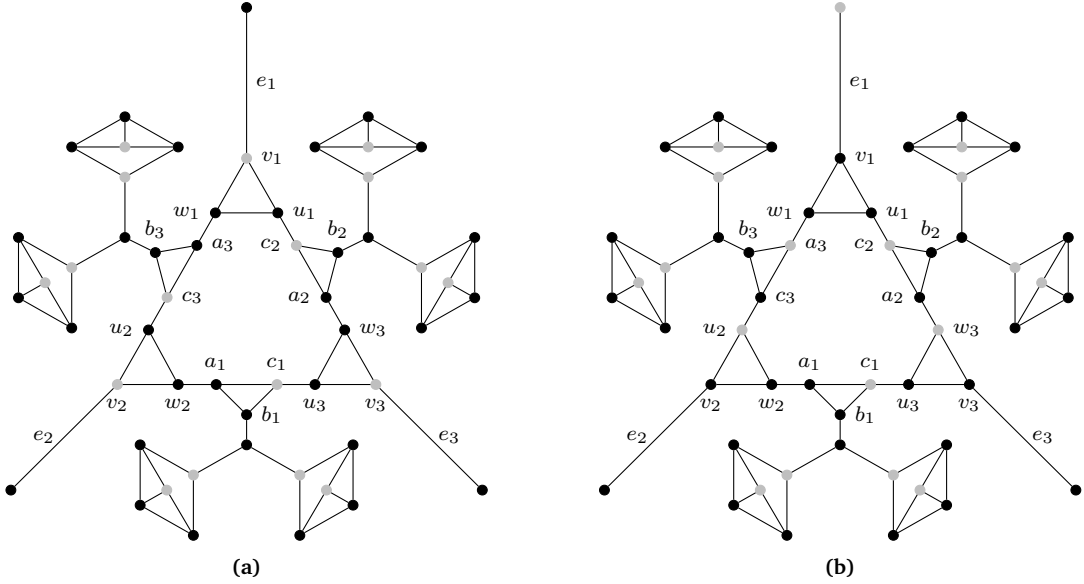 

Conversely, let $D'$ be a semi-TD-set of $G'$. We first observe the following:

\begin{observation}
\label{obs:si2}
\begin{minipage}[t]{\linegoal}
\begin{itemize}[leftmargin=*]
\item[$\bullet$] $D'$ contains at least two vertices of each subgraph of $G_v$ of the form \raisebox{-.5ex}{%
\begin{tikzpicture}[scale=0.3]
\node[cir] (1) at (0,0) {};
\node[cir] (2) at (1,0) {};
\node[cir] (3) at (1,.5) {};
\node[cir] (4) at (1,-.5) {};
\node[cir] (5) at (2,0) {};
\draw (1) -- (2)
(1) -- (3)
(1) -- (4)
(2) -- (3)
(2) -- (5)
(3) -- (5)
(4) -- (5);
\end{tikzpicture}}.
\item[$\bullet$] For each $i\in \{1,2,3\}$, at least one vertex of $S_i = N[b_i]$ belongs to $D'$.
\end{itemize}
\end{minipage}
\end{observation}

We next prove the following claim:

\begin{claim}\label{claimcountcubic}
$|D' \cap V(G_{v})| \geq 17$, for each $v \in V(G)$. Moreover, equality holds only if $D' \cap \{v_{1}, v_{2}, v_{3}\} = \varnothing$. 
\end{claim}

Let us first show that $|D' \cap V(G_{v})| \geq 17$, for each $v \in V(G)$. If $D' \cap \{u_{i}, v_{i}, w_{i}\} \neq \varnothing$, for each $i \in \{1,2,3\}$, then Observation \ref{obs:si2} implies that $|D' \cap V(G_{v})| \geq 3 + 12 + 3 = 18$. Suppose now  that there exists $i \in \{1,2,3\}$ such that $D' \cap \{u_{i}, v_{i}, w_{i}\} = \varnothing$. We then have that $\{a_{i-1},c_{i+1}\} \subseteq D'$ (indices are taken modulo $3$), as $w_i$ and $u_i$ are dominated. Moreover, the witness of $a_{i-1}$ and that of $c_{i+1}$ are distinct (as $D' \cap \{u_{i}, v_{i}, w_{i}\} = \varnothing$) and belong to neither $S_i$ nor to any subgraph of the form \raisebox{-.5ex}{%
\begin{tikzpicture}[scale=0.3]
\node[cir] (1) at (0,0) {};
\node[cir] (2) at (1,0) {};
\node[cir] (3) at (1,.5) {};
\node[cir] (4) at (1,-.5) {};
\node[cir] (5) at (2,0) {};
\draw (1) -- (2)
(1) -- (3)
(1) -- (4)
(2) -- (3)
(2) -- (5)
(3) -- (5)
(4) -- (5);
\end{tikzpicture}}.
Observation \ref{obs:si2} then implies that $|D' \cap V(G_{v})| \geq 2 + 2 + 12 + 1 = 17$. 

It remains to show that if $D'$ contains a vertex in $\{v_{1}, v_{2}, v_{3}\}$, then $|D' \cap V(G_{v})| \geq 18$. By the previous paragraph, we may assume that there exists $i \in \{1,2,3\}$ such that $D' \cap \{u_{i}, v_{i}, w_{i}\} = \varnothing$. As mentioned above, this implies that $\{a_{i-1},c_{i+1}\} \subseteq D'$ and that the witness of $a_{i-1}$ and that of $c_{i+1}$ are distinct. Moreover, these witnesses are different from $v_{i-1}$, $v_{i+1}$, any vertex of $S_i$ and any vertex of a subgraph of the form \raisebox{-.5ex}{%
\begin{tikzpicture}[scale=0.3]
\node[cir] (1) at (0,0) {};
\node[cir] (2) at (1,0) {};
\node[cir] (3) at (1,.5) {};
\node[cir] (4) at (1,-.5) {};
\node[cir] (5) at (2,0) {};
\draw (1) -- (2)
(1) -- (3)
(1) -- (4)
(2) -- (3)
(2) -- (5)
(3) -- (5)
(4) -- (5);
\end{tikzpicture}}. Therefore, since $|D' \cap \{v_{1}, v_{2}, v_{3}\}| \geq 1$, Observation \ref{obs:si2} implies that $|D' \cap V(G_v)| \geq 2 + 2 + 12 + 1 + 1 = 18$. $\lozenge$

\smallskip

Denoting by $t$ the number of gadgets $G_{v}$ of $G'$ such that $|D' \cap V(G_{v})| \geq 18$, \Cref{claimcountcubic} implies that $|D'| \geq 18t + 17(n - t) = t + 17n$. Let now $D$ be the set of vertices $v$ of $G$ corresponding to those gadgets $G_{v}$ in $G'$ such that $|D' \cap V(G_{v})| \geq 18$. We claim that $D$ is a dominating set of $G$. Indeed, if $v \notin D$, then $|D' \cap V(G_{v})| = 17$ and \Cref{claimcountcubic} implies that $D' \cap \{v_{1}, v_{2}, v_{3}\} = \varnothing$. Therefore, at least one vertex of $\{v_{1}, v_{2}, v_{3}\}$ is dominated by a vertex in a different gadget and since this vertex is of the form $x_{i}$, for some $x$ adjacent to $v$, we have that $|D' \cap V(G_{x})| \geq 18$ and so $x \in D$. But then $\gamma(G) \leq |D| = t \leq |D'| - 17n$ and letting $D'$ to be a minimum semi-TD-set of $G'$, we obtain $\gamma(G) \leq |D| \leq \gamma_{t2}(G') - 17n$, which proves our initial claim.

\smallskip

We can finally show that the construction introduced above gives rise to an L-reduction. Indeed, since $G$ is cubic, we have that $n \leq 4\gamma(G)$ and so $\gamma_{t2}(G') = \gamma(G) + 17n \leq 69\gamma(G)$. Moreover, given a semi-TD-set $D'$ of $G'$, the previous paragraphs show how to construct in polynomial time a dominating set $D$ of $G$ such that $|D| \leq |D'| - 17n$. Therefore, $|D| - \gamma(G) \leq |D'| - 17n - \gamma(G) = |D'| - \gamma_{t2}(G')$, thus concluding the proof. 
\end{proof}

In the rest of this section, we show the $\mathsf{APX}$-completeness of \textsc{Semitotal Dominating Set} when restricted to subcubic bipartite graphs. For this purpose, it is natural to consider the effect of odd subdivisions on the semitotal domination number. Clearly, replacing an edge of a graph with a path of length $4$ increases the domination number by exactly one. A slightly more complicated analysis gives the following result on the effect of $5$-subdivisions for the  semitotal domination number: 

\begin{lemma}\label{subdsemi} Let $G$ be a graph and $uv \in E(G)$. If $G'$ is the graph obtained from $G$ by replacing the edge $uv$ with a path of length $6$, then $\gamma_{t2}(G') = \gamma_{t2}(G) + 2$. 
\end{lemma}

\begin{proof} Let $uw_{1}w_{2}w_{3}w_{4}w_{5}v$ be the path of length $6$ in $G'$ resulting from the subdivision of $uv$. Suppose first that $D$ is a minimum semi-TD-set of $G$. We build a semi-TD-set of $G'$ as follows (see \Cref{fig:subd}). If $\{u, v\} \subseteq D$, it is easy to see that $D \cup \{w_{2}, w_{4}\}$ is a semi-TD-set of $G'$ of size $\gamma_{t2}(G) + 2$. Suppose now $D$ contains exactly one of $u$ and $v$, say without loss of generality $u \in D$. If no vertex in $N_{G}(v) \setminus \{u\}$ belongs to $D$, then $D \cup \{w_{3}, w_{5}\}$ is clearly a semi-TD-set of $G'$ of size $\gamma_{t2}(G) + 2$. On the other hand, if there exists $w \in (N_{G}(v) \setminus \{u\}) \cap D$, we claim that $D \cup \{w_{2}, w_{5}\}$ is a semi-TD-set of $G'$ of size $\gamma_{t2}(G) + 2$. Indeed, it is clearly a dominating set. Moreover, $w_{2}$ is a witness for $u$ and $w_{5}$ is a witness for every vertex adjacent to $v$. Suppose finally $D \cap \{u, v\} = \varnothing$. It is easy to see that $D \cup \{w_{2}, w_{4}\}$ is a semi-TD-set of $G'$ of size $\gamma_{t2}(G) + 2$.     

\begin{figure}[h!]
\centering
\begin{tikzpicture}[node distance=.6cm]
\node[circr,label=below:{\footnotesize $u$}] (u1) at (0,0) {};
\node[circr,label=below:{\footnotesize $v$}] (v1) at (1,0) {};
\node[circg,below left of=u1] (x1) {};
\node[circ,above left of=u1] (y1) {};
\node[circ,below right of=v1] (z1) {};
\node[circg,above right of=v1] (t1) {};
\draw (u1) -- (v1)
(u1) -- (x1)
(u1) -- (y1)
(v1) -- (z1)
(v1) -- (t1);

\draw[-Implies,line width=.6pt,double distance=2pt] (2.5,0) -- (3.5,0);

\node[circr,label=below:{\footnotesize $u$}] (u'1) at (5,0) {};
\node[circr] (w1) at (5.4,0) {};
\node[circg,label=above:{\footnotesize $w_2$}] (w2) at (5.8,0) {};
\node[circr] (w3) at (6.2,0) {};
\node[circg,label=above:{\footnotesize $w_4$}] (w4) at (6.6,0) {};
\node[circr] (w5) at (7,0) {};
\node[circr,label=below:{\footnotesize $v$}] (v'1) at (7.4,0) {};
\node[circg,below left of=u'1] (a1) {};
\node[circ,above left of=u'1] (b1) {};
\node[circ,below right of=v'1] (c1) {};
\node[circg,above right of=v'1] (d1) {};
\draw (u'1) -- (a1)
(u'1) -- (b1) 
(u'1) -- (w1)
(w1) -- (w2)
(w2) -- (w3)
(w3) -- (w4) 
(w4) -- (w5) 
(w5) -- (v'1)
(v'1) -- (c1)
(v'1) -- (d1); 

\node[circg,label=below:{\footnotesize $u$}] (u12) at (0,2) {};
\node[circr,label=below:{\footnotesize $v$}] (v12) at (1,2) {};
\node[circ,below left of=u12] (x12) {};
\node[circ,above left of=u12] (y12) {};
\node[circ,below right of=v12] (z12) {};
\node[circg,above right of=v12] (t12) {};
\draw (u12) -- (v12)
(u12) -- (x12)
(u12) -- (y12)
(v12) -- (z12)
(v12) -- (t12);

\draw[-Implies,line width=.6pt,double distance=2pt] (2.5,2) -- (3.5,2);

\node[circg,label=below:{\footnotesize $u$}] (u'12) at (5,2) {};
\node[circr] (w12) at (5.4,2) {};
\node[circg,label=above:{\footnotesize $w_2$}] (w22) at (5.8,2) {};
\node[circr] (w32) at (6.2,2) {};
\node[circr] (w42) at (6.6,2) {};
\node[circg,label=above:{\footnotesize $w_5$}] (w52) at (7,2) {};
\node[circr,label=below:{\footnotesize $v$}] (v'12) at (7.4,2) {};
\node[circ,below left of=u'12] (a12) {};
\node[circ,above left of=u'12] (b12) {};
\node[circ,below right of=v'12] (c12) {};
\node[circg,above right of=v'12] (d12) {};
\draw (u'12) -- (a12)
(u'12) -- (b12) 
(u'12) -- (w12)
(w12) -- (w22)
(w22) -- (w32)
(w32) -- (w42) 
(w42) -- (w52) 
(w52) -- (v'12)
(v'12) -- (c12)
(v'12) -- (d12); 

\node[circg,label=below:{\footnotesize $u$}] (u14) at (0,4) {};
\node[circr,label=below:{\footnotesize $v$}] (v14) at (1,4) {};
\node[circ,below left of=u14] (x14) {};
\node[circ,above left of=u14] (y14) {};
\node[circr,below right of=v14] (z14) {};
\node[circr,above right of=v14] (t14) {};
\draw (u14) -- (v14)
(u14) -- (x14)
(u14) -- (y14)
(v14) -- (z14)
(v14) -- (t14);

\draw[-Implies,line width=.6pt,double distance=2pt] (2.5,4) -- (3.5,4);

\node[circg,label=below:{\footnotesize $u$}] (u'14) at (5,4) {};
\node[circr] (w14) at (5.4,4) {};
\node[circr] (w24) at (5.8,4) {};
\node[circg,label=above:{\footnotesize $w_3$}] (w34) at (6.2,4) {};
\node[circr] (w44) at (6.6,4) {};
\node[circg,label=above:{\footnotesize $w_5$}] (w54) at (7,4) {};
\node[circr,label=below:{\footnotesize $v$}] (v'14) at (7.4,4) {};
\node[circ,below left of=u'14] (a14) {};
\node[circ,above left of=u'14] (b14) {};
\node[circr,below right of=v'14] (c14) {};
\node[circr,above right of=v'14] (d14) {};
\draw (u'14) -- (a14)
(u'14) -- (b14) 
(u'14) -- (w14)
(w14) -- (w24)
(w24) -- (w34)
(w34) -- (w44) 
(w44) -- (w54) 
(w54) -- (v'14)
(v'14) -- (c14)
(v'14) -- (d14); 

\node[circg,label=below:{\footnotesize $u$}] (u16) at (0,6) {};
\node[circg,label=below:{\footnotesize $v$}] (v16) at (1,6) {};
\node[circ,below left of=u16] (x16) {};
\node[circ,above left of=u16] (y16) {};
\node[circ,below right of=v16] (z16) {};
\node[circ,above right of=v16] (t16) {};
\draw (u16) -- (v16)
(u16) -- (x16)
(u16) -- (y16)
(v16) -- (z16)
(v16) -- (t16);

\draw[-Implies,line width=.6pt,double distance=2pt] (2.5,6) -- (3.5,6);

\node[circg,label=below:{\footnotesize $u$}] (u'16) at (5,6) {};
\node[circr] (w16) at (5.4,6) {};
\node[circg,label=above:{\footnotesize $w_2$}] (w26) at (5.8,6) {};
\node[circr] (w36) at (6.2,6) {};
\node[circg,label=above:{\footnotesize $w_4$}] (w46) at (6.6,6) {};
\node[circr] (w56) at (7,6) {};
\node[circg,label=below:{\footnotesize $v$}] (v'16) at (7.4,6) {};
\node[circ,below left of=u'16] (a16) {};
\node[circ,above left of=u'16] (b16) {};
\node[circ,below right of=v'16] (c16) {};
\node[circ,above right of=v'16] (d16) {};
\draw (u'16) -- (a16)
(u'16) -- (b16) 
(u'16) -- (w16)
(w16) -- (w26)
(w26) -- (w36)
(w36) -- (w46) 
(w46) -- (w56) 
(w56) -- (v'16)
(v'16) -- (c16)
(v'16) -- (d16);

\node[draw=none] at (.5,7) {$G$};
\node[draw=none] at (6.2,7) {$G'$}; 
\end{tikzpicture}
\caption{Construction of a semi-TD-set $D'$ of $G'$ from a semi-TD-set $D$ of $G$. The gray vertices belong to $D$, the red vertices do not belong to $D$ and the black vertices are undetermined. Similarly for $D'$. Note that we might have $N_{G}(u) \cap N_{G}(v) \neq \varnothing$.}
\label{fig:subd}
\end{figure}

Suppose now $D'$ is a minimum semi-TD-set of $G'$. We first claim that by eventually adding one of $u$ or $v$ to $D' \setminus \{w_{1}, w_{2}, w_{3}, w_{4}, w_{5}\}$ we obtain a set $D$ which is a semi-TD-set of $G$. Clearly, $D$ is a dominating set of $G$, as $u$ and $v$ are the only vertices of $G$ that can be dominated by a vertex $w_{i}$ (in particular, it must be $i \in \{1, 5\}$). Moreover, if a vertex $w \in D' \setminus \{w_{1}, w_{2}, w_{3}, w_{4}, w_{5}\}$ had a witness with respect to $D'$ in $\{w_{1}, w_{2}, w_{3}, w_{4}, w_{5}\}$, then either $wu \in E(G)$, $wv \in E(G)$, $w = u$ or $w = v$. In either case, $u$ or $v$ is a witness for $w$ with respect to $D$. Finally, it is easy to see that the eventually newly added vertex has a witness with respect to $D$. 
   
Observe now that, since $D'$ is a dominating set of $G'$, it contains at least two vertices of $\{w_{1}, w_{2}, w_{3}, w_{4}, w_{5}\}$. If it contains at least three such vertices then, by eventually adding $u$ or $v$ to $D' \setminus \{w_{1}, w_{2}, w_{3}, w_{4}, w_{5}\}$, we obtain a semi-TD-set of $G$ of size at most $\gamma_{t2}(G') - 2$. We may therefore assume that $D'$ contains exactly two vertices of $\{w_{1}, w_{2}, w_{3}, w_{4}, w_{5}\}$. 

If $w_{3} \in D'$, then either $u$ or $v$ belongs to $D'$, for otherwise $D'$ must contain at least three vertices of $D$, and it is easy to see that $D' \setminus \{w_{1}, w_{2}, w_{3}, w_{4}, w_{5}\}$ is a semi-TD-set of $G$ of size $\gamma_{t2}(G') - 2$. It remains to consider the case $w_{3} \notin D'$. Suppose first exactly one of $w_{2}$ and $w_{4}$ belongs to $D'$, say without loss of generality, $w_{2} \in D'$. Since $w_{4} \notin D'$, we have that $w_{5} \in D'$. Moreover, since $w_{2}$ has a witness with respect to $D'$, it must be that $u \in D'$. Since $w_5$ has a witness not in $\{w_{1}, w_{2}, w_{3}, w_{4}, w_{5}\}$, it is easy to see that $D' \setminus \{w_{1}, w_{2}, w_{3}, w_{4}, w_{5}\}$ is a semi-TD-set of $G$ of size $\gamma_{t2}(G') - 2$. Suppose finally both $w_{2}$ and $w_{4}$ belong to $D'$. By the reasonings above, we may assume that $D'$ contains at most one of $u$ and $v$. If $D'$ contains exactly one of $\{u, v\}$, say without loss of generality $u \in D'$, then a vertex not in $\{w_{1}, w_{2}, w_{3}, w_{4}, w_{5}\}$ dominates $v$ and it is easy to see that $D' \setminus \{w_{1}, w_{2}, w_{3}, w_{4}, w_{5}\}$ is a semi-TD-set of $G$ of size $\gamma_{t2}(G') - 2$. On the other hand, if $\{u, v\} \cap D' = \varnothing$, then $u$ is dominated by a vertex in $D' \setminus \{w_{1}, w_{2}, w_{3}, w_{4}, w_{5}\}$ and $v$ is dominated by a vertex in $D' \setminus \{w_{1}, w_{2}, w_{3}, w_{4}, w_{5}\}$. Therefore, $D' \setminus \{w_{1}, w_{2}, w_{3}, w_{4}, w_{5}\}$ is a dominating set of $G$ of size $\gamma_{t2}(G') - 2$. Moreover, no vertex in $\{w_{1}, w_{2}, w_{3}, w_{4}, w_{5}\}$ is a witness for a vertex in $D' \setminus \{w_{1}, w_{2}, w_{3}, w_{4}, w_{5}\}$  with respect to $D'$ and so $D' \setminus \{w_{1}, w_{2}, w_{3}, w_{4}, w_{5}\}$ is indeed a semi-TD-set of $G$. 
\end{proof}

\begin{corollary}\label{bipsub} \textsc{Semitotal Dominating Set} is $\mathsf{APX}$-complete when restricted to subcubic bipartite graphs.
\end{corollary}

\begin{proof} We construct an L-reduction from \textsc{Semitotal Dominating Set} restricted to cubic graphs, which is $\mathsf{APX}$-complete by \Cref{semicubic}. Given an instance $G$ of this problem with $n$ vertices and $m$ edges, we build a subcubic bipartite graph $G'$ by taking a $5$-subdivision of $G$. By \Cref{subdsemi}, $\gamma_{t2}(G') = \gamma_{t2}(G) + 2m = \gamma_{t2}(G) + 3n$. Since $n \leq 4\gamma(G) \leq 4\gamma_{t2}(G)$, we have that $\gamma_{t2}(G') \leq 13\gamma_{t2}(G)$. Moreover, by the proof of \Cref{subdsemi}, given a semi-TD-set $D'$ of $G'$ we can find in polynomial time a semi-TD-set $D$ of $G$ such that $|D| \leq |D'| - 2m$. Therefore, $|D| - \gamma_{t2}(G) \leq |D'| - 2m - \gamma_{t2}(G) = |D'| - \gamma_{t2}(G')$, thus concluding the proof. 
\end{proof}

\section{The reduction from \textsc{Semitotal Dominating Set} to \textsc{Total Dominating Set}}\label{reduc}

In this short section, we introduce a graph transformation providing a polynomial-time reduction from \textsc{Semitotal Dominating Set} to \textsc{Total Dominating Set}. In \Cref{mim,dually} we will then observe that the class of graphs with bounded mim-width and the class of dually chordal graphs are both closed under this transformation. We remark that a similar strategy has been adopted in \citep{KS97} in order to solve \textsc{Total Dominating Set} for graph classes closed under addition of false twins by reducing it to \textsc{Dominating Set}. 

Given a graph $G = (V, E)$, we construct the transformed graph $G' = (V', E')$ as follows. The vertex set $V'$ of $G'$ consists of two disjoint copies $V_1$ and $V_2$ of $V$. For a vertex $v \in V$ and $i \in \{1, 2\}$, we denote by $v_i$ the copy of $v$ in $V_i$. A vertex $v_1 \in V_1$ is adjacent in $G'$ to every vertex $u_2 \in V_2$ such that $u \in N_G[v]$. In addition, a vertex $v_2 \in V_2$ is adjacent to every vertex $u_2 \in V_2$ such that $u \in N_G^2(v)$. In other words, $G'$ is obtained from $G$ by first adding a true twin for each vertex of $G$ and then adding edges between vertices of $G$ at distance $2$. In particular, $G'[V_2]$ is isomorphic to $G^2$. Clearly, $G'$ can be constructed in $O(\vert V \vert \cdot \vert E \vert)$ time.

We now show that we can efficiently obtain a minimum semitotal dominating set of $G$ from a minimum total dominating set of $G'$:

\begin{lemma}\label{hat}
For any graph $G$, we have $\gamma_{t2}(G) = \gamma_t(G')$. Moreover, a minimum semitotal dominating set of $G$ can be obtained from a minimum total dominating set of $G'$ in linear time.
\end{lemma}

\begin{proof}
Consider a minimum semitotal dominating set $S$ of $G$. We claim that $S' = \{v_2 \in V_2: v \in S\}$ is a total dominating set of $G'$. Since $G'[V_2]$ is isomorphic to $G^2$, $S'$ is dominating in $G'[V_2]$ and every vertex in $S'$ has a neighbour in $S'$, as $S$ is a semitotal dominating set. Consider now $v_1 \in V_1$. Either $v \in S$ or there exists $u \in S$ dominating $v$ in $G$. In the former case, $v_1$ is dominated by $v_2$ in $G'$, while in the latter $v_1u_2 \in E'$ and $u_2 \in S'$ by construction. Therefore, $S'$ is a total dominating set of $G'$ and $\gamma_t(G') \leq \gamma_{t2}(G)$.

Conversely, let $S'$ be a minimum total dominating set of $G'$. Observe first that we can assume, without loss of generality, that $v_1$ and $v_2$ are not both in $S'$. Indeed, since $N_{G'}[v_1] \subseteq N_{G'}[v_2]$, if both $v_1$ and $v_2$ are in $S'$, then $(S' \setminus \{v_1\}) \cup \{u_2\}$ with $u_2 \in N_{G'}(v_1) \setminus \{v_2\}$ is a minimum total dominating set of $G'$.  

We now claim that $S = \{v \in V: \{v_1,v_2\} \cap S' \neq \varnothing\}$ is a dominating set of $G$. Indeed, consider $v \in V$. If $\{v_1,v_2\} \cap S' \neq \varnothing$, then $v \in S$. Otherwise, there exists $u_2\in S'$ distinct from $v_{2}$ and dominating $v_1$ in $G'$. But then $u \in S$ and $u$ dominates $v$ in $G$. We finally claim that each vertex in $S$ is within distance two of another. Indeed, consider $v \in S$. If $v_1 \in S'$, there exists $u_2 \in S'$, with $u_2 \neq v_2$ by assumption, such that $u_2v_1 \in E'$. This implies that $u \in S$ and $uv \in E$. On the other hand, if $v_2 \in S'$, there exists $w' \in S'$ distinct from $v_1$ and such that $w'v_{2} \in E'$. But $w'$ is the copy of a vertex $w \in V$ which by construction belongs to $S$ and is at distance at most $2$ from $v$. Therefore, $S$ is a semitotal dominating set of $G$ and $\gamma_{t2}(G) \leq \gamma_t(G')$, thus concluding the proof of the first assertion. 

Consider now the second assertion and let $S'$ be a minimum total dominating set of $G'$. As long as there exist $v_1$ and $v_2$ both in $S'$, we update $S' = (S' \setminus \{v_1\}) \cup \{u_2\}$, for some $u_2 \in N_{G'}(v_1) \setminus \{v_2\}$, and return the set $S = \{v \in V: \{v_1,v_2\} \cap S' \neq \varnothing\}$. By the previous paragraphs, $S$ is a minimum semitotal dominating set of $G$.    
\end{proof}

\section{Graphs of bounded mim-width}\label{mim}

The maximum induced matching width (mim-width for short) is a graph parameter introduced by \citet{Vat12}
measuring how easy it is to decompose a graph along vertex cuts inducing a bipartite graph with small maximum induced matching size. The modelling power of mim-width is stronger than that of tree-width and clique-width, in the sense that graphs with bounded clique-width have bounded mim-width but there exist graph classes (interval graphs and permutation graphs) with mim-width $1$ \citep{Vat12} and unbounded clique-width \citep{GR00}. 

Mim-width has important consequences for the so-called $(\sigma, \rho)$-domination problems, a subclass of graph problems expressible in $\mbox{MSO}_{1}$ introduced by \citet{TP97} as follows. Given two finite or co-finite subsets $\sigma$ and $\rho$ of $\mathbb{N}$ and a graph $G$, a vertex set $S \subseteq V(G)$ is a \textit{$(\sigma, \rho)$-dominator} if the following two conditions hold:
\begin{itemize}
\item $|N(v) \cap S| \in \sigma$, for each $v \in S$;
\item $|N(v) \cap S| \in \rho$, for each $v \in V(G)\setminus S$.
\end{itemize}

For example, a $(\{0\}, \mathbb{N})$-dominator is an independent set, a $(\mathbb{N}, \mathbb{N}^{+})$-dominator is a
dominating set and a $(\mathbb{N}^{+}, \mathbb{N}^{+})$-dominator is a total dominating set. An algorithmic problem is a \textit{$(\sigma, \rho)$-domination problem} if the property in question can be described by a $(\sigma, \rho)$-dominator. 

Combining results in \citep{BV13,BTV13}, it is known that the three versions of a $(\sigma, \rho)$-domination problem (minimisation, maximisation, existence) can be solved in $O(n^{w})$ time, assuming a decomposition tree with mim-width $w$ is provided as part of the input. It should be remarked that deciding the mim-width of a graph is $\mathsf{NP}$-hard in general and not in $\mathsf{APX}$ unless $\mathsf{NP} = \mathsf{ZPP}$ \citep{SV16}. However, \citet{BV13} showed that it is possible to find decomposition trees of constant mim-width in polynomial time for the following classes of graphs: permutation graphs, convex graphs and their complements, interval graphs and their complements, (circular $k$-) trapezoid graphs, circular permutation graphs, Dilworth-$k$ graphs, $k$-polygon graphs, circular arc graphs and complements of $d$-degenerate graphs.

\citet{JKST18} showed that the distance-$r$ version of a $(\sigma, \rho)$-domination problem (i.e. the version obtained by replacing $N(v)$ with $N^{r}(v)$) can be polynomially reduced to the $(\sigma, \rho)$-domination problem. The key ingredient for their reduction is the fact that, for any positive integer $r$, the mim-width of the graph power $G^{r}$ is at most twice the mim-width of $G$. Note that \textsc{Semitotal Dominating Set} is not a distance-$r$ $(\sigma, \rho)$-domination problem but it is in some sense a combination of a distance-$1$ $(\mathbb{N}, \mathbb{N}^{+})$-domination problem with a distance-$2$ $(\mathbb{N}^{+}, \mathbb{N}^{+})$-domination problem.  

\citet{HP17} asked for the complexity of \textsc{Semitotal Dominating Set} for the following two subclasses of chordal bipartite graphs: bipartite permutation graphs and convex bipartite graphs. In this section, we answer their question and in fact prove a more general result: \textsc{Semitotal Dominating Set} can be solved in $O(n^{w})$ time, assuming a decomposition tree with mim-width $w$ is provided as part of the input. We use the transformation introduced in \Cref{reduc} and show that, for any graph $G$ not isomorphic to $tK_{1}$, the mim-width of $G'$ is at most twice that of $G$. Combined with the results mentioned above, this immediately implies the claimed polynomial-time algorithm.  

In order to prove our results, we first properly define the notion of mim-width. A \textit{decomposition tree} for a graph $G$ is a pair $(T, \delta)$, where $T$ is a subcubic tree and $\delta$ is a bijection between the vertices of $G$ and the leaves of $T$. Each edge $e \in E(T)$ naturally splits the leaves of the tree in two groups depending on their component when $e$ is removed. In this way, each edge $e \in E(T)$ also represents a partition of $V(G)$ into two partition classes $A_{e}$ and $\overline{A_{e}}$, denoted by $(A_{e}, \overline{A_{e}})$. Denoting by $G[X, Y]$ the bipartite subgraph of $G$ induced by the edges with one endpoint in $X$ and the other in $Y$, we define the mim-width of $G$ as follows:

\begin{definition} Let $G$ be a graph and $(T, \delta)$ a decomposition tree for $G$. For each edge $e \in E(T)$ and the corresponding partition of the vertices $(A_{e}, \overline{A_{e}})$, we denote by $\textit{cutmim}_{G}(A_{e}, \overline{A_{e}})$ the size of a maximum induced matching in $G[A_{e}, \overline{A_{e}}]$. The mim-width of the
decomposition tree $(T, \delta)$ is the quantity $\textit{mimw}_{G}(T, \delta) = \max_{e \in E(T)}\textit{cutmim}_{G}(A_{e}, \overline{A_{e}})$. The mim-width $\textit{mimw}(G)$ of the graph $G$ is the minimum value of $\textit{mimw}_{G}(T, \delta)$ over all possible decompositions trees $(T, \delta)$ for $G$. 
\end{definition}

As mentioned above, the mim-width of $G^{2}$ is at most twice the mim-width of $G$. We now show that the same bound holds for the transformed graph $G'$ introduced in \Cref{dually}:   

\begin{theorem}\label{mimG'} For any graph $G$ not isomorphic to $tK_{1}$, $\textit{mimw}(G') \leq 2 \cdot \textit{mimw}(G)$.  
\end{theorem}

\begin{proof} We build a decomposition tree $(T', \delta')$ for $G'$ by growing a decomposition tree $(T, \delta)$ for $G$ as follows. Recall first that $V(G') = V_{1} \cup V_{2}$, where each $v_{i} \in V_{i}$ corresponds to $v \in V(G)$. Denoting by $x_{v}$ the image of $v \in V(G)$ under $\delta$, we add the vertices $x_{v_{1}}$ and $x_{v_{2}}$ to $V(T)$ and the edges $x_{v}x_{v_{1}}$ and $x_{v}x_{v_{2}}$ to $E(T)$. Repeating this procedure for each leaf $x_{v}$ of $T$, we clearly obtain a decomposition tree for $G'$. We now show that the resulting $(T', \delta')$ has mim-width at most $2 \cdot \textit{mimw}_{G}(T, \delta)$. This would conclude the proof.  

Consider an edge $e'$ of $T'$ and the corresponding partition $(A_{e'}, \overline{A_{e'}})$ such that $\textit{cutmim}_{G'}(A_{e'}, \overline{A_{e'}})$ attains the maximum over all edges of $T'$, i.e. $\textit{mimw}_{G'}(T', \delta') = \textit{cutmim}_{G'}(A_{e'}, \overline{A_{e'}})$, and let $M'$ be a maximum induced matching in $G'[A_{e'}, \overline{A_{e'}}]$. If $e'$ is incident to a leaf of $T'$, then $|M'| = 1$, and since there exists $v \in V(G)$ such that $d_{G}(v) \geq 1$, we have $\textit{mimw}_{G}(T, \delta) \geq 1$. Therefore, by construction, we may assume that $e' \in E(T)$. Consider now the cut $(B_{e'}, \overline{B_{e'}})$ of $G$ obtained by removing $e'$ from $T$. It is easy to see that it can be obtained from $(A_{e'}, \overline{A_{e'}})$ by replacing each pair of vertices $v_{1}, v_{2}$ with $v$. In the following, we simply denote $(A_{e'}, \overline{A_{e'}})$ and $(B_{e'}, \overline{B_{e'}})$ by $(A, \overline{A})$ and $(B, \overline{B})$, respectively, and we show how to construct an induced matching $M$ in $G[B, \overline{B}]$ such that $|M| \geq |M'|/2$. 

We first claim that there exists a maximum induced matching in $G'[A, \overline{A}]$ containing no edge $u_{2}v_{2}$ such that $d_{G}(u, v) = 2$. Indeed, let $M'$ be a maximum induced matching in $G'[A, \overline{A}]$ containing the minimum number $m'$ of such edges. Suppose to the contrary that $m' > 0$. Let $u_{2}v_{2}$ be an edge as above and consider a vertex $w \in V(G)$ such that $d_{G}(u, w) = d_{G}(w, v) = 1$. Without loss of generality, $\{w_{1}, w_{2}\}$ belongs to the same partition class as $\{v_{1}, v_{2}\}$. We show that $M'' = (M' \setminus \{u_{2}v_{2}\}) \cup \{u_{2}w_{1}\}$ is an induced matching in $G'[A, \overline{A}]$, thus obtaining a contradiction. Observe first that $w_{1}$ is not matched in $M'$. Indeed, every neighbour of $w_{1}$ in $G'$ is also a neighbour of $v_{2}$. Suppose now, to the contrary, that $M''$ is not an induced matching in $G'[A, \overline{A}]$. Since $M'$ is an induced matching, $w_{1}$ is adjacent to some vertex $y_{2} \neq u_{2}$ (in the partition class containing $u_{2}$) which is matched in $M'$. But then $y_{2}v_{2} \in E(G')$, contradicting the fact that $M'$ is induced. 

Therefore, there exists a maximum induced matching $M'$ in $G'[A, \overline{A}]$ containing no edge $u_{2}v_{2}$ such that $d_{G}(u, v) = 2$. We begin the construction of $M$ by adding all edges $uv \in E(G)$ such that $u_{2}v_{2} \in M'$. Note that, since $M'$ is induced, $u_{1}$ and $v_{1}$ are not matched in $M'$. Let now $N' \subseteq M'$ be the subset containing edges of the form $u_{1}v_{2}$, with $u_{1} \in V_{1}$ and $v_{2} \in V_{2}$. Note that, for each $u_{1}v_{2} \in N'$, $d_{G}(u, v) = 1$. Moreover, since $M'$ is induced, $u_{2}$ and $v_{1}$ are not matched in $M'$. Without loss of generality, we may assume that at least half of the edges in $N'$ are such that their endpoints in $V_{2}$ belong to $A$. We then add to $M$ the set $\{uv : u_{1}v_{2} \in N' \ \mbox{and} \ v_{2} \in A\}$. 

Clearly, $|M| \geq |M'|/2$ and so it remains to show that $M'$ is an induced matching in $G[B, \overline{B}]$. Therefore, consider $\{uv, st\} \subseteq M$. If $u_{2}v_{2} \in M'$ and $s_{2}t_{2} \in M'$ then, without loss of generality, $v_{2}$ and $t_{2}$ belong to the same partition class. Since $M'$ is induced, $u_{2}t_{2} \notin E(G')$ and $s_{2}v_{2} \notin E(G')$ and so $d_{G}(u, t) > 2$ and $d_{G}(s, v) > 2$. On the other hand, if $u_{1}v_{2} \in M'$ and $s_{2}t_{2} \in M'$ then, without loss of generality, $v_{2}$ and $t_{2}$ belong to the same partition class. This implies that $u_{1}t_{2} \notin E(G')$ and $s_{2}v_{2} \notin E(G')$ and so $d_{G}(u, t) > 1$ and $d_{G}(s, v) > 2$. Finally, if $u_{1}v_{2} \in M'$ and $s_{1}t_{2} \in M'$ then, by construction, $t_{2}$ and $v_{2}$ belong to the same partition class. Therefore, $u_{1}t_{2} \notin E(G')$ and $s_{1}v_{2} \notin E(G')$ and so $d_{G}(u, t) > 1$ and $d_{G}(s, v) > 1$. This concludes the proof. 
\end{proof}

The following result was obtained by rephrasing earlier results in \citep{BV13,BTV13}:

\begin{theorem}[\citet{JKST18}]\label{total} There is an algorithm that, given a graph $G$ and a decomposition tree $(T, \delta)$ of $G$ with $w = \textit{mimw}_{G}(T, \delta)$, solves \textsc{Total Dominating Set} in $O(n^{4 + 3w})$ time.
\end{theorem}

Combining \Cref{mimG'} and the algorithm above, we can immediately obtain a polynomial-time algorithm for \textsc{Semitotal Dominating Set}: 

\begin{theorem}\label{boundmimsemi} There is an algorithm that, given a graph $G$ and a decomposition tree $(T, \delta)$ of $G$ with $w = \textit{mimw}_{G}(T, \delta)$, solves \textsc{Semitotal Dominating Set} in $O(n^{4 + 6w})$ time.
\end{theorem}

\begin{proof} Let $G$ be the input graph and $(T, \delta)$ the given decomposition tree with $w = \textit{mimw}_{G}(T, \delta)$. Clearly, we may assume that $G$ is not isomorphic to $tK_{1}$. The algorithm proceeds as follows. We first compute the graph $G'$ and the decomposition tree $(T', \delta')$ of $G'$ as in \Cref{mimG'}. We then run the algorithm given by \Cref{total} to solve \textsc{Total Dominating Set} in $G'$ with decomposition tree $(T', \delta')$ and apply \Cref{hat} in order to find a minimum semitotal dominating set of $G$ from a minimum total dominating set of $G'$. 

Correctness is guaranteed by \Cref{hat} and \Cref{mimG'}. As for the running time, computing $G'$ and $(T', \delta')$ from $G$ and $(T, \delta)$ can be done in $O(n^{3})$ time. Moreover, by the proof of \Cref{mimG'}, $\textit{mimw}_{G'}(T', \delta') \leq 2w$ and so, by \Cref{total}, the second step can be executed in $O(n^{4 + 6w})$ time.
\end{proof}

\section{Dually chordal graphs}\label{dually}

In this section we show that the class of dually chordal graphs is closed under the transformation introduced in \Cref{reduc}. In view of \Cref{hat}, in order to provide a polynomial-time algorithm for \textsc{Semitotal Dominating Set} restricted to dually chordal graphs, it would then be enough to apply a polynomial-time algorithm for \textsc{Total Dominating Set}  restricted to that class. \citet{KS97} claimed such an algorithm by a reduction to \textsc{Dominating Set}. Their reasoning is based on the fact that if $G$ is dually chordal, then the lexicographic product $G[2K_{1}]$ is dually chordal as well. However, $P_{3}$ is dually chordal but $P_{3}[2K_{1}]$ is not and so it remains an open problem to determine whether \textsc{Total Dominating Set} and \textsc{Semitotal Dominating Set} are polynomial-time solvable for dually chordal graphs. Note that dually chordal graphs do not have bounded mim-width. In fact, \citet{Men17} showed that the classes of strongly chordal split graphs, co-comparability graphs and circle graphs do not have bounded mim-width: there exist infinite subfamilies of such graph classes having mim-width bounded from below linearly in the number of vertices.  

Before embarking on the task above, we recall some definitions. 

\begin{definition}[\citet{BDCV98}]\label{disk}
A graph $G$ is dually chordal if there exists a spanning tree $T$ of $G$ such that any disk of $G$ induces a subtree in $T$. 
\end{definition}

Observe that every graph with a dominating vertex is dually chordal and so the class of dually chordal graphs is not hereditary and it is not a subclass of perfect graphs. Moreover, the previous observation readily implies that many $\mathsf{NP}$-hard graph problems remain $\mathsf{NP}$-hard for dually chordal graphs: for example \textsc{Vertex Cover}, \textsc{Colouring}, $k$-\textsc{Colouring} for each $k \geq 4$ (however, $3$-\textsc{Colouring} is in $\mathsf{P}$ \citep{Lei17}), \textsc{Clique} and \textsc{Clique Cover}. On the other hand, some variants of \textsc{Dominating Set} become polynomial-time solvable when restricted to this class \citep{BCD98,BFLM15}.

Dually chordal graphs have been introduced as a generalisation of strongly chordal graphs. In order to define this graph class, we need the following notion. For $n \geq 3$, the \textit{$n$-sun} $S_{n}$ is the graph on $2n$ vertices whose vertex set can be partitioned into an independent set $W = \{w_{1}, \dots, w_{n}\}$ and a clique $U = \{u_{1}, \dots, u_{n}\}$ and such that, for each $i$ and $j$, $w_{i}$ is adjacent to $u_{j}$ if and only if $i = j$ or $i \equiv j + 1 \pmod n$. 

\begin{definition}[\citet{Far83}] A graph is strongly chordal if it is chordal and sun-free, i.e. $S_{n}$-free for each $n \geq 3$. 
\end{definition}

It turns out that strongly chordal graphs are exactly the hereditary dually chordal graphs:

\begin{theorem}[\citet{BDCV98}]
A graph $G$ is strongly chordal if and only if every induced subgraph of $G$ is dually chordal.
\end{theorem}

\citet{HP17} showed that \textsc{Semitotal Dominating Set} can be solved in $O(n^{2})$ time for interval graphs and asked for the complexity of the problem restricted to strongly chordal graphs. We observe that forbidding only a finite number of induced suns does not make the problem easy:

\begin{theorem}\label{split} For any $k \geq 3$, \textsc{Semitotal Dominating Set} is $\mathsf{NP}$-complete for $(S_{3}, \dots, S_{k})$-free split graphs.  
\end{theorem}

For the proof of \Cref{split}, we use the following construction introduced in \citep{Ale03}. Given a graph $G$, let $Q(G)$ be the graph with vertex set $V(G) \cup E(G)$ and edge set $$\{v_{i}v_{j} : v_{i}, v_{j} \in V(G)\} \cup \{ve : v \in V(G), \ e \in E(G), \ v \in e\}.$$ Clearly, $Q(G)$ is a connected split graph and each vertex in the independent set $E(G)$ has degree $2$. Moreover, the following holds (see, e.g., \citep{MP15}):

\begin{lemma}\label{splitred} For every graph $G$, we have $\gamma(Q(G)) = \beta(G)$.
\end{lemma}

Note that \Cref{splitred} is stated in \citep{MP15} with the additional assumption of $G$ being connected, but it is easy to see it holds for any graph $G$.

\begin{lemma}[Folklore]\label{folk} For every connected split graph $G$ with no dominating vertex, we have $\gamma(G) = \gamma_{t2}(G) = \gamma_{t}(G)$.
\end{lemma}

\begin{proof} Clearly, it is enough to show that $\gamma_{t}(G) \leq \gamma(G)$. Let $V(G) = C \cup I$ be a partition of the vertices of $G$ into a clique $C$ and an independent set $I$. Note that there exists a minimum dominating set of $G$ contained in $C$. Indeed, if a minimum dominating set $D$ of $G$ contains $u \in I$, then no neighbour of $u$ belongs to $D$ and so, denoting by $v$ one such vertex, we have that $(D \setminus \{u\}) \cup \{v\}$ is a minimum dominating set containing less vertices of $I$. Since $G$ has no dominating vertex, every dominating set contained in $C$ is a total dominating set and so $\gamma_{t}(G) \leq \gamma(G)$. 
\end{proof}

We can finally prove \Cref{split}: 

\begin{proof}[Proof of \Cref{split}] We reduce from \textsc{Vertex Cover} restricted to $(C_{3}, \dots, C_{k})$-free graphs, which is known to be $\mathsf{NP}$-complete (see, e.g., \citep{Ale03}). Note that the problem remains $\mathsf{NP}$-hard even if the instances do not contain dominating vertices, as can be easily seen by taking a $2$-subdivision. 

Given now an instance $G$ of this problem, we construct the split graph $Q(G)$ introduced above. Suppose that $Q(G)$ contains an induced $S_{i}$, for some $i \leq k$, and let $V(S_{i}) = W \cup U$ be a partition as in the definition of $i$-sun. Since the vertices of $Q(G)$ can be partitioned into a clique $V(G)$ and an independent set $E(G)$, we have that $U \subseteq V(G)$ and $W \subseteq E(G)$. Therefore, the vertices corresponding to $U$ form a cycle in $G$ (not necessarily induced) of length $i$, a contradiction. 

Since $G$ has no dominating vertex, $Q(G)$ has no dominating vertex either. Moreover, since $Q(G)$ is connected, \Cref{folk,splitred} imply $\gamma_{t2}(Q(G)) = \gamma(Q(G)) = \beta(G)$ and the conclusion follows.
\end{proof}

It is easy to see that the proof of \Cref{split} can be adapted to show the following:

\begin{theorem}\label{splittot} For any $k \geq 3$, \textsc{Dominating Set} and \textsc{Total Dominating Set} are $\mathsf{NP}$-complete for $(S_{3}, \dots, S_{k})$-free split graphs.  
\end{theorem}

We now come back to dually chordal graphs and finally show that this class is closed under the transformation introduced in \Cref{reduc}. In order to do so, we first observe that to check the spanning tree property in \Cref{disk}, it is enough to consider the closed neighbourhoods. For ease of exposition, we use the following terminology: a \textit{compatible tree} for a graph $G$ is a spanning tree $T$ of $G$ such that $N_{G}[z]$ induces a subtree in $T$, for any $z \in V(G)$.

\begin{lemma}\label{compatible} A graph $G$ is dually chordal if and only if there exists a compatible tree for $G$.
\end{lemma}

\begin{proof} We only show that if there exists a compatible tree for $G$ then $G$ is dually chordal, as the other implication is trivial. Therefore, suppose there exists a spanning tree $T$ of $G$ such that $N_{G}[z]$ induces a subtree in $T$, for any $z \in V(G)$. Let $z$ be a fixed arbitrary vertex of $G$. We claim that, for any $r \geq 1$, $N_{G}^r[z]$ induces a subtree in $T$ and proceed by induction on $r$, the case $r = 1$ being true by assumption. 

Suppose now that $N_{G}^{r}[z]$ induces a subtree in $T$ and consider $v \in N_{G}^{r+1}[z]$. It is enough to show that $z$ and $v$ are connected in $T$ by a path consisting of vertices of $N_{G}^{r+1}[z]$. Let $v'$ be the neighbour of $v$ in a shortest $z, v$-path in $G$. We have that $v' \in N_{G}^{r}[z]$ and, by the induction hypothesis, $v'$ is connected in $T$ to $z$ by a path consisting of vertices of $N_{G}^{r}[z]$. Moreover, $v$ is connected in $T$ to $v'$ by a path consisting of vertices of $N_{G}[v']$ and since these vertices belong to $N_{G}^{r+1}[z]$, the conclusion immediately follows.     
\end{proof}

\begin{theorem}\label{duallysemi} If $G$ is a dually chordal graph, then the transformed graph $G'$ is.  
\end{theorem}

\begin{proof}
Let $G = (V, E)$ be a dually chordal graph and $G' = (V', E')$ the transformed graph. Since $G$ is dually chordal, there exists a spanning tree $T$ of $G$ such that any disk of $G$ induces a subtree in $T$. Moreover, for any $z \in V$, $N_{G^2}[z] = N_G^{2}[z]$ and so $T$ is a compatible tree for $G^2$. We now construct from $T$ a spanning tree $T'$ of $G'$ as follows. For any $v \in V$, $v_1v_2 \in E(T')$ and for any $u$ and $v$ in $V$, $u_2v_2 \in E(T')$ if and only if $uv \in E(T)$. 

We claim that $T'$ is a compatible tree for $G'$, i.e. $N_{G'}[z_{i}]$ induces a subtree in $T'$, for each $z_{i} \in V'$. By \Cref{compatible}, this would conclude the proof. Suppose first that $i = 1$. Since $N_{G'}(z_{1}) = \{v_{2} \in G' : v \in N_{G}[z]\}$ and $T$ is a compatible tree for $G$, we have that $N_{G'}(z_{1})$ induces a subtree in $T'$. Moreover, $z_{2} \in N_{G'}(z_{1})$ and $z_1z_2 \in E(T')$ and so $N_{G'}[z_{1}]$ induces a subtree in $T'$.  

Suppose finally that $i = 2$. Since $T$ is a compatible tree for $G^2$, every vertex in $N_{G'}[z_{2}]\setminus V_{1}$ is connected in $T$ (and so in $T'$) to $z_{2}$ by a path consisting of vertices of $N_{G'}[z_{2}] \setminus V_{1}$. Moreover, if $v_{1} \in N_{G'}[z_{2}] \cap V_{1}$, then $v_{2} \in N_{G'}[z_{2}]\setminus V_{1}$ and $v_{1}v_{2} \in E(T')$. Therefore, $N_{G'}[z_{2}]$ induces a subtree in $T'$. 
\end{proof}

\section{Complexity dichotomies in monogenic classes}\label{sec:dicho}

Recall that a class of graphs $\mathcal{G}$ is monogenic if it is defined by a single forbidden induced subgraph, i.e. $\mathcal{G} = \mbox{\textit{Free}}(H)$, for some graph $H$. We say that a (decision) graph problem admits a dichotomy in monogenic classes if, for each monogenic class, the problem is either $\mathsf{NP}$-complete or decidable in polynomial time. \citet{Kor92} showed that \textsc{Dominating Set} is decidable in polynomial time if $H$ is an induced subgraph of $P_{4} + tK_{1}$, for $t \geq 0$, and $\mathsf{NP}$-complete otherwise. Dichotomies for other problems have been provided in~\citep{Abou18,GPS14,Kam12,KKTW01}. 

In this section, combining results of \Cref{secapp} with existing ones, we obtain complexity dichotomies for \textsc{Semitotal Dominating Set} and the closely related \textsc{Total Dominating Set}. Observe first that the proof of \Cref{inapprox} and the fact that \textsc{Dominating Set} is $\mathsf{NP}$-complete for subcubic graphs \citep{GJ79} immediately imply the following:

\begin{corollary}\label{semidec} \textsc{Semitotal Dominating Set} is $\mathsf{NP}$-complete for subcubic line graphs of bipartite graphs.
\end{corollary}

Moreover, by applying a $5$-subdivision sufficiently many times, \Cref{semidec} implies the following:   

\begin{corollary}\label{largegirth} For any $k \geq 2$, \textsc{Semitotal Dominating Set} is $\mathsf{NP}$-complete for subcubic bipartite $(C_{4}, \dots, C_{2k})$-free graphs.  
\end{corollary}

We have seen in \Cref{prel} that if a graph property is expressible in $\mbox{MSO}_{1}$, then it is decidable in polynomial time for graphs of bounded clique-width. We now show that being a dominating set, a total dominating set or a semitotal dominating set are examples of such a property:

\begin{lemma}[Folklore]\label{mso1} Being a dominating set, a total dominating set or a semitotal dominating set are all expressible in $\textsc{MSO}_{1}$.
\end{lemma}

\begin{proof} Let $G = (V, E)$ be a graph. The following $\mbox{MSO}_{1}$ sentence says that $D \subseteq V$ is a dominating set:
\begin{equation*} \textbf{dom}(D) = \forall_{v \in V \setminus D}\exists_{u \in D} \ \textit{adj}(u, v).
\end{equation*}

The following $\mbox{MSO}_{1}$ sentence says that $D \subseteq V$ is a total dominating set:
\begin{equation*} \textbf{semi-TD-set}(D) = \textbf{dom}(D) \wedge \forall_{u \in D}\exists_{v \in D\setminus\{u\}} \ \textit{adj}(u, v).
\end{equation*}

Finally, the following $\mbox{MSO}_{1}$ sentence says that $D \subseteq V$ is a semitotal dominating set:
\begin{equation*} \textbf{semi-TD-set}(D) = \textbf{dom}(D) \wedge (\forall_{u \in D}\exists_{v \in D\setminus\{u\}} \ \textit{adj}(u, v) \vee (\exists_{x \in V} \ \textit{adj}(u, x) \wedge \textit{adj}(x, v))).
\end{equation*}
\end{proof}

We can finally prove the complexity dichotomy for \textsc{Semitotal Dominating Set}: 

\begin{theorem}\label{dich} \textsc{Semitotal Dominating Set} restricted to $H$-free graphs is decidable in polynomial time if $H$ is an induced subgraph of $P_{4} + tK_{1}$ and $\mathsf{NP}$-complete otherwise.
\end{theorem}

\begin{proof} If $H$ contains an induced cycle $C_{k}$, then the problem is $\mathsf{NP}$-complete by \Cref{largegirth}. Moreover, if $H$ is a forest with a vertex of degree at least $3$, then $H$ contains an induced claw and the problem is $\mathsf{NP}$-complete by \Cref{inapprox}.

Finally, suppose $H$ is the disjoint union of paths. If $H$ contains at least two paths on at least $2$ vertices, then $H$ contains $2K_{2}$ and the problem is $\mathsf{NP}$-complete since it is $\mathsf{NP}$-complete when restricted to split graphs \citep{HP17} (see also \Cref{split}). The same conclusion holds if $H$ contains a path on at least $5$ vertices. It remains to consider the case of $H$ being of the form $P_{k} + tK_{1}$, for some $k \leq 4$ and $t \geq 0$. Therefore, let $G$ be such a $P_{k} + tK_{1}$-free graph. If $G$ is in addition $P_{k}$-free, then it has bounded clique-width and the problem is decidable in polynomial time (\Cref{mso1}). On the other hand, if $G$ contains an induced copy $G'$ of $P_{k}$, then there are at most $t - 1$ pairwise non-adjacent vertices of $G$ none of which is adjacent to a vertex of $G'$. Denoting this set by $S$, we have that $V(G') \cup S$ is a dominating set of size at most $t + 3$. Moreover, denoting by $\gamma(G)$ the domination number of $G$, it is easy to see that $\gamma_{t2}(G) \leq 2\gamma(G)$ and so $\gamma_{t2}(G) \leq 2t + 6$. Therefore, it suffices to check all the subsets of $V(G)$ of size at most $2t + 6$, and this can be clearly done in polynomial time.     
\end{proof}

We now consider \textsc{Total Dominating Set} and show that it remains $\mathsf{NP}$-hard for graphs with arbitrarily large girth:  

\begin{lemma}\label{girthtot} Let $G$ be a graph and $uv \in E(G)$. If $G'$ is the graph obtained from $G$ by replacing the edge $uv$ with a path of length $5$, then $\gamma_{t}(G') = \gamma_{t}(G) + 2$.

In particular, for any $k \geq 3$, \textsc{Total Dominating Set} is $\mathsf{NP}$-complete for $(C_{3}, \dots, C_{k})$-free graphs.   
\end{lemma}

\begin{proof} Let $uw_{1}w_{2}w_{3}w_{4}v$ be the path of length $5$ in $G'$ resulting from the subdivision of $uv$. Suppose first that $D$ is a minimum TD-set of $G$. We build a TD-set of $G'$ as follows. If $\{u, v\} \subseteq D$, it is easy to see that $D \cup \{w_{1}, w_{4}\}$ is a TD-set of $G'$ of size $\gamma_{t}(G) + 2$. If exactly one of $u$ and $v$ belongs to $D$, say without loss of generality $u \in D$, then $D \cup \{w_{3}, w_{4}\}$ is a TD-set of $G'$ of size $\gamma_{t}(G) + 2$. Finally, if $\{u,v\} \cap D = \varnothing$, it is easy to see that $D \cup \{w_2,w_3\}$ is a TD-set of $G'$ of size $\gamma_{t}(G) + 2$.     

Suppose now that $D'$ is a minimum TD-set of $G'$. Clearly, $D'$ contains at least two vertices of $\{w_{1}, w_{2}, w_{3}, w_{4}\}$. If $\{u, v\} \subseteq D'$, we have that $D' \setminus \{w_{1}, w_{2}, w_{3}, w_{4}\}$ is a TD-set of $G$ of size at most $\gamma_{t}(G') - 2$. 

Suppose now that exactly one of $u$ and $v$ belongs to $D'$, say without loss of generality $u \in D'$. This implies that $w_3$ belongs to $D'$, as $w_4$ has a neighbour in $D'$ (independently of whether it belongs to $D'$ or not). Since $w_3$ has a neighbour in $D'$, we have $|D' \cap \{w_{2}, w_{3}, w_{4}\}| \geq 2$. Therefore, if $w_{1} \in D'$, then $(D' \cup \{v\}) \setminus \{w_{1}, w_{2}, w_{3}, w_{4}\}$ is a TD-set of $G$ of size at most $\gamma_{t}(G') - 2$. On the other hand, if $w_1 \notin D'$, then $u$ has a neighbour in $D'$ different from $w_1$, and so $D' \setminus \{w_{1}, w_{2}, w_{3}, w_{4}\}$ is a TD-set of $G$ of size at most $\gamma_{t}(G') - 2$. 

Finally, suppose that $\{u,v\} \cap D' = \varnothing$. If $\{w_1,w_4\} \subseteq D'$, then $\{w_1,w_2,w_3,w_4\} \subseteq D'$ and so $D' \setminus \{w_1,w_2,w_3,w_4\} \cup \{u, v\}$ is a TD-set of $G$ of size $\gamma_{t}(G') - 2$. If exactly one of $w_1$ and $w_4$ belongs to $D'$, say without loss of generality $w_1 \in D'$, then $\{w_2, w_3\} \subseteq D'$. Therefore, $D' \setminus \{w_1,w_2,w_3,w_4\} \cup \{v\}$ is a TD-set of $G$ of size $\gamma_t(G') - 2$. Finally, if $\{w_1,w_4\} \cap D' = \varnothing$, then $\{w_2,w_3\} \subseteq D'$ and $D' \setminus \{w_1,w_2,w_3,w_4\}$ is a TD-set of $G$ of size $\gamma_t(G') - 2$.

The second statement immediately follows by applying a $4$-subdivision to an instance of \textsc{Total Dominating Set} sufficiently many times.
\end{proof}

Using the fact that \textsc{Total Dominating Set} is $\mathsf{NP}$-complete for graphs with arbitrarily large girth (\Cref{girthtot}), for split graphs \citep{LP83} (see also \Cref{splittot}) and for claw-free graphs \citep{Mc94}, the reasonings in the proof of \Cref{dich} immediately imply the following dichotomy:

\begin{theorem}\label{dichtot} \textsc{Total Dominating Set} restricted to $H$-free graphs is decidable in polynomial time if $H$ is an induced subgraph of $P_{4} + tK_{1}$ and $\mathsf{NP}$-complete otherwise.
\end{theorem}

Comparing \Cref{dich,dichtot} and results in \citep{Kor92,Mun3}, it can be observed that the complexities of \textsc{Dominating Set}, \textsc{Semitotal Dominating Set}, \textsc{Total Dominating Set} and \textsc{Connected Dominating Set} all agree when restricted to monogenic classes. Therefore, differences in the complexities of these problems restricted to hereditary classes might arise only when forbidding more than one induced subgraph. For example, \textsc{Total Dominating Set} is in $\mathsf{P}$ for chordal bipartite graphs \citep{DMK90}, whereas \textsc{Semitotal Dominating Set} is $\mathsf{NP}$-hard for this class \citep{HP17}. \citet{HP17} provided an example of a graph class for which \textsc{Semitotal Dominating Set} is in $\mathsf{P}$ but \textsc{Total Dominating Set} is $\mathsf{NP}$-hard. However, their example is not hereditary and so it is natural to ask: 

\begin{question} Does there exist a hereditary class for which \textsc{Semitotal Dominating Set} is in $\mathsf{P}$ but \textsc{Total Dominating Set} is $\mathsf{NP}$-hard?
\end{question}

Moreover, looking at \Cref{table}, it is tempting to ask:

\begin{question} If \textsc{Dominating Set} is in $\mathsf{P}$ for a hereditary class $\mathcal{C}$, is it true that \textsc{Semitotal Dominating Set} is in $\mathsf{P}$ for $\mathcal{C}$ as well? 
\end{question}

\section{Perfectness}\label{sec:perfect}

Since $\gamma(G) \leq \gamma_{t2}(G) \leq \gamma_{t}(G)$, for any graph $G$ with no isolated vertex, it is natural to ask for which graphs equalities hold. In this section, we show that the graphs attaining either of the two equalities are unlikely to have a polynomial characterisation. Indeed, for a given graph $G$, we show that it is $\mathsf{NP}$-hard to decide whether $\gamma_{t2}(G) = \gamma_{t}(G)$, even if $G$ is planar and with maximum degree at most $4$. Our reduction is based on the following problem, shown to be $\mathsf{NP}$-complete by \citet{DJPS94}.

\begin{center}
\fbox{%
\begin{minipage}{5.5in}
\textsc{Planar Exactly $3$-Bounded $3$-Sat}
\begin{description}[\compact\breaklabel\setleftmargin{60pt}]
\item[Instance:] A formula $\Phi$ with variable set $X$ and clause set $C$ such that each variable has exactly three literals with one of them occuring in two clauses and the other in one, and each clause is the disjunction of either two or three literals. Moreover, the bipartite graph $G_{X, C}$ with vertex set $X \cup C$ and an edge between $x \in X$ and $c \in C$ if $c$ contains one of the literals $x$ and $\overline{x}$ is planar.
\item[Question:] Is $\Phi$ satisfiable?   
\end{description}
\end{minipage}}
\end{center}

\begin{theorem}\label{t2t} For a given planar graph $G$ with maximum degree at most $4$, it is $\mathsf{NP}$-hard to decide whether $\gamma_{t2}(G) = \gamma_{t}(G)$.
\end{theorem}

\begin{proof} As mentioned above, we reduce from \textsc{Planar Exactly $3$-Bounded $3$-Sat}. Given a formula $\Phi$ with variable set $X$ and clause set $C$, we build a graph $G$ as follows. For any variable $x \in X$, we introduce a triangle $G_{x}$ with two distinguished \textit{literal vertices} $x$ and $\overline{x}$. For any clause $c \in C$, we introduce a copy of $K_{2}$, denoted by $G_{c}$, with a distinguished \textit{clause vertex} $c$. Finally, for each clause $c \in C$, we add an edge between the clause vertex $c$ and the literal vertices whose corresponding literals belong to $c$. Clearly, $G$ is planar and has maximum degree at most $4$. We claim that $\Phi$ is satisfiable if and only if $\gamma_{t2}(G) = \gamma_{t}(G)$.

Observe first that $\gamma_{t2}(G) = |X| + |C|$. Indeed, let $D$ be a dominating set of $G$. For each $x \in X$, we have that $|D \cap V(G_{x})| \geq 1$ and, for each $c \in C$, $|D \cap V(G_{c})| \geq 1$. Therefore, $\gamma_{t2}(G) \geq |X| + |C|$. Moreover, the set consisting of all clause vertices together with an arbitrary vertex for each variable gadget is a semitotal dominating set of $G$ of size $|X| + |C|$.

Suppose now that $\Phi$ is satisfiable. Given a satisfying truth assignment of $\Phi$, we build a total dominating set $T$ of $G$ of size $|X| + |C|$ as follows. For any variable $x \in X$, we add to $T$ either the literal vertex $x$, if $x$ evaluates to true, or the literal vertex $\overline{x}$, otherwise. Moreover, for any clause $c \in C$, we add the clause vertex $c$ to $T$. Clearly, $|T| = |X| + |C|$ and $T$ is a dominating set of $G$. Since each clause vertex belongs to $T$, we have that every literal vertex in $T$ has a neighbour in $T$. Moreover, since each clause vertex is adjacent to a literal vertex whose corresponding literal evaluates to true, we have that $T$ is indeed a total dominating set of $G$.

Conversely, suppose that $\gamma_{t2}(G) = \gamma_{t}(G)$ and let $T$ be a total dominating set of $G$ such that $|T| = \gamma_{t}(G)$. By assumption, $|T| = |X| + |C|$. This implies that, for each $c \in C$, $T \cap V(G_{c}) = \{c\}$ and, for each $x \in X$, either $T \cap V(G_{x}) = \{x\}$ or $T \cap V(G_{x}) = \{\overline{x}\}$. Therefore, we define a truth assignment of $\Phi$ as follows: we set $x$ to true if $T \cap V(G_{x}) = \{x\}$, and to false if $T \cap V(G_{x}) = \{\overline{x}\}$. Since each clause vertex belongs to $T$ and is adjacent to a literal vertex which belongs to $T$ as well, we have that this truth assignment indeed satisfies $\Phi$.
\end{proof}

\citet{ADR15} showed that, for a given graph $G$, it is $\mathsf{NP}$-hard to decide whether $\gamma(G) = \gamma_{t}(G)$. As implicitly observed in the proof of \Cref{t2t}, the graph $G$ therein satisfies $\gamma(G) = |X| + |C|$ and so we immediately obtain the following strengthening of their result: 

\begin{corollary}\label{domtotdom} For a given planar graph $G$ with maximum degree at most $4$, it is $\mathsf{NP}$-hard to decide whether $\gamma(G) = \gamma_{t}(G)$.
\end{corollary}

We now show that, similarly to \Cref{t2t} and \Cref{domtotdom}, recognising the graphs $G$ such that $\gamma(G) = \gamma_{t2}(G)$ is $\mathsf{NP}$-hard.

\begin{theorem}\label{domsemidom} For a given planar graph $G$ with maximum degree at most $4$, it is $\mathsf{NP}$-hard to decide whether $\gamma(G) = \gamma_{t2}(G)$.
\end{theorem}

\begin{proof} We reduce again from \textsc{Planar Exactly $3$-Bounded $3$-Sat}. Given a formula $\Phi$ with variable set $X$ and clause set $C$, we build a graph $G$ as follows. For any variable $x \in X$, we introduce the gadget $G_{x}$ depicted in \Cref{fig:gad}. It contains two distinguished \textit{literal vertices} $x$ and $\overline{x}$. Moreover, if $x \in X$ has one negated and two unnegated occurences in $\Phi$, then the literal vertex $\overline{x} \in V(G_{x})$ is adjacent to the \textit{literal transmitter} $\overline{x}_{1}$ of $\overline{x}$, whereas the literal vertex $x \in V(G_{x})$ is adjacent to the two \textit{literal transmitters} $x_{1}$ and $x_{2}$ of $x$. We proceed, mutatis mutandis, if $x \in X$ has one unnegated and two negated occurences in $\Phi$. Then, for any clause $c \in C$, we introduce a copy of $K_{2}$, denoted by $G_{c}$, with a distinguished \textit{clause vertex} $c$. Finally, for each clause $c \in C$, we add an edge between the clause vertex $c$ and the literal transmitters whose corresponding literals belong to $c$. Clearly, $G$ is planar and has maximum degree at most $4$. We claim that $\Phi$ is satisfiable if and only if $\gamma(G) = \gamma_{t2}(G)$.

\begin{figure}[h!]
\centering
\begin{tikzpicture}
\node[circ,label=left:{\footnotesize $x_1$}] (x1) at (0,0) {};
\node[circ,label=below:{\footnotesize $x$}] (x) at (1,0) {};
\node[circ,label=below:{\footnotesize $\overline{x}$}] (xb) at (3,0) {};
\node[circ,label=right:{\footnotesize $\overline{x}_1$}] (x1b) at (4,0) {};
\node[circ] (a) at (1,1) {};
\node[circ,label=below:{\footnotesize $t_x$}] (tx) at (2,1) {};
\node[circ] (b) at (3,1) {};
\node[circ] (c) at (2,2) {};
\node[circ,below left of=x,label=below:{\footnotesize $x_2$}] (x2) {}; 
\draw (x1) -- (x) 
(x2) -- (x)
(x) -- (xb)
(x) -- (a) 
(a) -- (tx)
(tx) -- (c)
(tx) -- (b)
(b) -- (xb)
(xb) -- (x1b);

\node[circ,label=left:{\footnotesize $\overline{x}_1$}] (x1bis) at (6,0) {};
\node[circ,label=below:{\footnotesize $\overline{x}$}] (xbis) at (7,0) {};
\node[circ,label=below:{\footnotesize $x$}] (xbb) at (9,0) {};
\node[circ,label=right:{\footnotesize $x_1$}] (x1bb) at (10,0) {};
\node[circ] (ab) at (7,1) {};
\node[circ,label=below:{\footnotesize $t_x$}] (txb) at (8,1) {};
\node[circ] (bb) at (9,1) {};
\node[circ] (cb) at (8,2) {};
\node[circ,below left of=xbis,label=below:{\footnotesize $\overline{x}_2$}] (x2b) {}; 
\draw (x1bis) -- (xbis) 
(x2b) -- (xbis)
(xbis) -- (xbb)
(xbis) -- (ab) 
(ab) -- (txb)
(txb) -- (cb)
(txb) -- (bb)
(bb) -- (xbb)
(xbb) -- (x1bb);
\end{tikzpicture}
\caption{The variable gadget $G_{x}$. If $x \in X$ has two unnegated occurences, then $G_{x}$ is depicted on the left. If $x \in X$ has two negated occurences, then $G_{x}$ is depicted on the right.}
\label{fig:gad}
\end{figure}
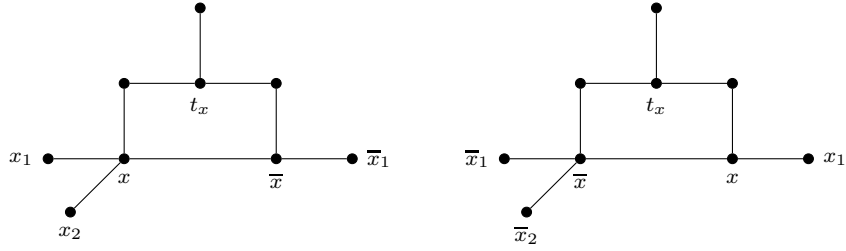

We first show that $\gamma(G) = 2|X| + |C|$. Indeed, let $D$ be a dominating set of $G$. For each $x \in X$, we have that $|D \cap V(G_{x})| \geq 2$ and, for each $c \in C$, $|D \cap V(G_{x})| \geq 1$. Therefore, $\gamma(G) \geq 2|X| + |C|$. Moreover, it is easy to see that $\{c : c \in C\} \cup \{t_{x}: x \in X\} \cup \{x: x \in X\}$ is a dominating set of $G$ of size $2|X| + |C|$.

Suppose now that $\Phi$ is satisfiable. Given a satisfying truth assignment of $\Phi$, we build a semitotal dominating set $S$ of $G$ of size $2|X| + |C|$ as follows. For any variable $x \in X$, we add to $S$ the vertex $t_{x}$ and either the literal vertex $x$, if $x$ evaluates to true, or the literal vertex $\overline{x}$, otherwise. Moreover, for any clause $c \in C$, we add the clause vertex $c$ to $S$. Clearly, $|S| = 2|X| + |C|$ and $S$ is a dominating set of $G$. Note that each vertex in $S \cap V(G_{x})$ is within distance $2$ of another vertex in $S \cap V(G_{x})$ and each clause vertex is within distance $2$ of a literal vertex in $S$ whose corresponding literal evaluates to true. Therefore, $S$ is indeed a semitotal dominating set of $G$.

Conversely, suppose that $\gamma(G) = \gamma_{t2}(G)$ and let $S$ be a semitotal dominating set of $G$ such that $|S| = \gamma_{t2}(G)$. By assumption, $|S| = 2|X| + |C|$. Since $S$ is a dominating set of $G$ and $|S| = 2|X| + |C|$, we have that $|S \cap V(G_{x})| = 2$, for each $x \in X$, and $|D \cap V(G_{x})| = 1$, for each $c \in C$. We now claim that $S \cap V(G_{c}) = \{c\}$, for each $c \in C$. Indeed, if this is not the case, $S$ contains the neighbour of $c$ in $G_{c}$ and a literal transmitter in some variable gadget $G_{x}$. But then it is easy to see that $|S \cap V(G_{x})| \geq 3$, a contradiction. Moreover, $t_{x} \in S$, for each $x \in X$, for otherwise its $1$-neighbour belongs to $S$ as well as one of its $2$-neighbours, thus implying that $|S \cap V(G_{x})| \geq 3$, a contradiction. Finally, since for each $x \in X$, $t_{x} \in S$ and $|S \cap V(G_{x})| = 2$, we have that either $S \cap V(G_{x}) = \{t_{x}, x\}$ or $S \cap V(G_{x}) = \{t_{x}, \overline{x}\}$. Therefore, we define a truth assignment of $\Phi$ as follows: we set $x$ to true if $x \in S \cap V(G_{x})$, and to false if $\overline{x} \in S \cap V(G_{x})$. Since each clause vertex belongs to $S$ and is within distance $2$ of a literal vertex which belongs to $S$, we have that the truth assignment defined above indeed satisfies $\Phi$.
\end{proof}

In view of \Cref{t2t,domsemidom}, it makes sense to study the class of graphs obtained by further requiring equality for every induced subgraph. The most prominent example in this respect is given by the class of perfect graphs\footnote{Recall that a graph $G$ is perfect if $\chi(H) = \omega(H)$, for every induced subgraph $H$ of $G$.}: \citet{CRST06} showed that a graph is perfect if and only if it does not contain any odd hole nor odd antihole. Graphs for which the equality between certain domination parameters holds for every induced subgraph received considerable attention. For example, \citet{Zve03} showed that a graph has equal domination number and connected domination number, for each of its connected induced subgraphs, if and only if it is $(P_{5}, C_{5})$-free. \citet{CP17} provided a forbidden induced subgraph characterisation for the graphs having equal domination number and independent domination number for every induced subgraph. Characterisations involving other domination parameters have been obtained in~\citep{ADR15,HH18,Sch12}.

In the rest of this section, we provide forbidden induced subgraph characterisations for the graphs heriditarily satisfying $\gamma_{t2} = \gamma_{t}$ and $\gamma = \gamma_{t2}$. These characterisations have been independently obtained by \citet{HH18}. However, our proofs have the nice feature of being considerably shorter.

\begin{theorem} For every graph $G$, the following are equivalent:
\begin{itemize}
\item For any induced subgraph $H$ of $G$ with no isolated vertex, $\gamma_{t2}(H) = \gamma_{t}(H)$. 
\item $G$ is $(C_5,P_5)$-free.
\end{itemize}
\end{theorem}

\begin{proof}
If the equality holds for any induced subgraph of $G$ with no isolated vertex, it is clear that $G$ is $(C_5, P_5)$-free as $\gamma_t(C_5) = \gamma_t(P_5) = 3$ and $\gamma_{t2} (C_5) = \gamma_{t2} (P_5) = 2$. 

Conversely, given an induced subgraph $H$ with no isolated vertex of a $(C_{5}, P_{5})$-free graph $G$, we show that there exists a minimum semitotal dominating set of $H$ which is a total dominating set of $H$. This would imply that $\gamma_{t}(H) \leq \gamma_{t2}(H)$ and thus conclude the proof. We say that a vertex $v$ in a semitotal dominating set of $H$ is \textit{bad} if $v$ has no witness at distance $1$ (in other words, every witness for $v$ is at distance exactly $2$). Suppose, to the contrary, that no minimum semitotal dominating set of $H$ is a total dominating set of $H$ and let $S$ be a minimum semitotal dominating set of $H$ having the minimum number of bad vertices. 

Consider a bad vertex $v \in S$ and let $w$ be a witness for $v$ with respect to $S$ (by assumption, $d_{H}(v, w) = 2$). We now show that $v$ and $w$ have incomparable neighbourhoods. Indeed, if $N_{H}(w) \subseteq N_{H}(v)$ then, since $N_{H}(v) \cap S = \varnothing$, the only vertex of $H$ witnessed by $w$ is $v$. Therefore, denoting by $u$ an arbitrary neighbour of $w$, we have that $(S \setminus \{w\}) \cup \{u\}$ is a semitotal dominating set of $H$ with fewer bad vertices than $S$, a contradiction. Hence, $N_{H}(w) \nsubseteq N_{H}(v)$. Similarly, it is easy to see that $N_{H}(v) \nsubseteq N_{H}(w)$. 

As shown in the previous paragraph, $v$ and $w$ have incomparable neighbourhoods. Consider now a vertex $y \in N_{H}(v) \cap N_{H}(w)$. We have that either $N_{H}(v) \setminus N_{H}(w)$ or $N_{H}(w) \setminus N_{H}(v)$ is complete to $y$, for otherwise there exist $v' \in N_{H}(v) \setminus N_{H}(w)$ and $w' \in N_{H}(w) \setminus N_{H}(v)$ both non-adjacent to $y$ and so $\{v', v, y, w, w'\}$ induces either a $C_{5}$ or a $P_{5}$, a contradiction. But if $N_{H}(v) \setminus N_{H}(w)$ (or $N_{H}(w) \setminus N_{H}(v)$) is complete to $y$, then $(S \setminus \{v\}) \cup \{y\}$ (or $(S \setminus \{w\}) \cup \{y\}$) is a semitotal dominating set of $H$ with fewer bad vertices than $S$, a contradiction. 
\end{proof}

We now turn to the graphs hereditarily satisfying $\gamma = \gamma_{t2}$. Note that if $G$ has no isolated vertex and contains a dominating vertex, then $\gamma(G) = 1$ but $\gamma_{t2}(G) = 2$. Therefore, it makes sense to modify the definition of $\gamma_{t2}$ as follows: 

\begin{equation*}
  \widetilde{\gamma}_{t2}(G)=
     \begin{cases}
        1 & \text{if $G$ has a dominating vertex;} \\
        \gamma_{t2}(G) & \text{otherwise.}
     \end{cases}
\end{equation*}

The graph $P$ depicted in \Cref{graphP} appears in the characterisation.

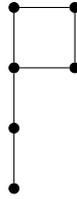
\begin{figure}[htb]
\centering
\begin{tikzpicture}[scale=.8]
\node[circ] (1) at (0,0) {};
\node[circ] (2) at (0,1) {};
\node[circ] (3) at (0,2) {};
\node[circ] (4) at (0,3) {};
\node[circ] (5) at (1,2) {};
\node[circ] (6) at (1,3) {};
\draw (1) edge[-] (2)
(2) edge[-] (3)
(3) edge[-] (4)
(3) edge[-] (5)
(4) edge[-] (6)
(5) edge[-] (6);
\end{tikzpicture}
\caption{The graph $P$.}
\label{graphP}
\end{figure}

\begin{theorem} For every graph $G$, the following are equivalent:
\begin{itemize}
\item For any induced subgraph $H$ of $G$ with no isolated vertex, $\gamma(H) = \widetilde{\gamma}_{t2}(H)$. 
\item $G$ is $(C_6,P_6,P)$-free.
\end{itemize}
\end{theorem}

\begin{proof}
If the equality holds for any induced subgraph of $G$ with no isolated vertex, then it is clear that $G$ is $(C_6,P_6,P)$-free as $\gamma (C_6) = \gamma (P_6) = \gamma (P) = 2$ and $\widetilde{\gamma}_{t2} (C_6) = \widetilde{\gamma}_{t2} (P_6) =\widetilde{\gamma}_{t2} (P) =3$. 

Conversely, given an induced subgraph $H$ with no isolated vertex of a $(C_6,P_6,P)$-free graph $G$, we show that there exists a minimum dominating set of $H$ which is also a semitotal dominating set of $H$. This would imply that $\widetilde{\gamma}_{t2}(H) \leq \gamma (H)$ and thus conclude the proof. Clearly, we may assume that $H$ has no dominating vertex. A vertex in a dominating set of $H$ is said to be \textit{bad} if it has no witness. Suppose, to the contrary, that no minimum dominating set of $H$ is a semitotal dominating set of $H$ and let $S$ be a minimum dominating set of $H$ having the minimum number of bad vertices. Note that $|S| > 1$, as $H$ has no dominating vertex. We now show that we can construct from $S$ a minimum dominating set with fewer bad vertices, thus contradicting the minimality of $S$ and concluding the proof. 

Consider a bad vertex $v \in S$. If $d_{H}(v) = 1$ then, denoting by $u$ the only neighbour of $v$ in $H$, the minimum dominating set $S' = (S\setminus \{v\}) \cup \{u\}$ has fewer bad vertices. Indeed, since $H$ has no dominating vertex, every neighbour of $u$ different from $v$ is dominated by a vertex in $S$ and so $u$ has a witness in $S'$. 

Assume henceforth that $d_{H}(v) > 1$ and consider a vertex $u \in S$ at minimum distance from $v$. We denote by $P$ a shortest path from $v$ to $u$. Since $v$ is bad, $N_H(v) \cap N_H(u) = \varnothing$, i.e. $d_H(v,u) > 2$. We now claim that $d_H(v,u) = 3$. Suppose, to the contrary, that $d_H(v,u) > 3$. Since no vertex of $P-\{v,u\}$ belongs to $S$, denoting by $w$ the vertex on $P$ at distance $2$ from $v$, there exists $a \in S \cap N_H(w)$ dominating $w$. But then $d_H(v,a) \leq 3$, which contradicts the definition of $u$.

Therefore $d_H(v,u) = 3$. Moreover $d_{H}(u) > 1$, for otherwise, denoting by $w$ the only neighbour of $u$ in $H$, the minimum dominating set $(S \setminus \{u\}) \cup \{w\}$ has fewer bad vertices than $S$. Let now $z$ and $y$ be the neighbours of $v$ and $u$ on $P$, respectively. We distinguish two cases according to whether or not there exists a neighbour of $v$ non-adjacent to both $z$ and $y$.

\bigskip

\noindent
\textbf{Case 1.} There exists $x \in N_H(v)$ non-adjacent to both $z$ and $y$. 

\noindent
Observe first that there exists $t \in N_H(u)$ non-adjacent to $y$. Indeed, $d_{H}(u) > 1$, and if every neighbour of $u$ is adjacent to $y$, then the minimum dominating set $(S\setminus \{u\}) \cup \{y\}$ of $H$ has fewer bad vertices than $S$, as $y$ is a witness for $v$ and for all the vertices previously witnessed by $u$, a contradiction. Observe now that $t$ is adjacent to both $z$ and $x$, for otherwise $\{x, v, z, y, u, t\}$ induces either a $P_{6}$, a $C_{6}$ or a $P$. Moreover, there exists $w \in N_H(u) \cup N_H(v)$ non-adjacent to both $z$ and $t$. Indeed, if every $w \in N_H(u) \cup N_H(v)$ is adjacent to either $z$ or $t$, then the minimum dominating set $(S\setminus \{u,v\}) \cup \{z,t\}$ of $H$ contains fewer bad vertices than $S$.

\begin{figure}[htb]
\centering
\begin{tikzpicture}[node distance=1cm]
\node[circ,label=below:{\footnotesize $v$}] (v) {};
\node[circ,label=below:{\footnotesize $z$},right of=v] (z) {};  
\node[circ,label=below:{\footnotesize $y$},right of=z] (y) {};
\node[circ,label=above:{\footnotesize $x$},above of=v] (x) {};
\node[circ,label=above:{\footnotesize $t$},right of=x] (t) {};
\node[circ,label=above:{\footnotesize $u$},right of=t] (u) {};

\draw (v) -- (z)
(v) -- (x)
(z) -- (y)
(z) -- (t)
(y) -- (u)
(x) -- (t)
(t) -- (u);
\end{tikzpicture}
\caption{Case 1.}
\end{figure}

The following two observations follow from the fact that $H$ is $(C_6, P_6, P)$-free.

\begin{equation}
\text{$w$ is adjacent to both $x$ and $y$.}
\label{c2}
\end{equation}

By symmetry, we may assume that $w$ is a neighbour of $v$. But then $w$ is adjacent to $y$, for otherwise $\{t, u, y, z, v, w\}$ induces a $P$, and so $w$ is adjacent to $x$ as well, or else $\{w,v,x,t,u,y\}$ induces a $C_6$. $\diamond$

\begin{equation}
\label{c1}
\text{Each vertex in } (N_H(u) \cup N_H(v)) \setminus \{x, y\} \text{ is adjacent to either $x$ or $y$.}
\end{equation} 

Suppose, to the contrary, that there exists $a \in (N_H(u) \cup N_H(v)) \setminus \{x, y\}$ adjacent to neither $x$ nor $y$. By symmetry, we may assume that $a$ is a neighbour of $u$. But then $a$ is adjacent to $z$, for otherwise $\{a, u, y, z, v, x\}$ induces a $P_6$, and so $\{a, u, y, z, v, x\}$ induces a $P$, a contradiction. $\diamond$ \\


 
\smallskip

By \labelcref{c2}, $x$ and $y$ are at distance two. It then follows from \labelcref{c1} that $(S\setminus \{u,v\}) \cup \{x,y\}$ is a minimum dominating set of $H$ with fewer bad vertices than $S$, a contradiction.

\bigskip

\noindent
\textbf{Case 2.} Every vertex in $N_{H}(v)$ is adjacent to either $z$ or $y$.

\noindent
If every neighbour of $u$ is adjacent to either $z$ or $y$, then the minimum dominating set $(S \setminus \{u, v\}) \cup \{z, y\}$ contains fewer bad vertices than $S$, a contradiction. Therefore, there exists $t \in N_H(u)$ non-adjacent to both $z$ and $y$ and we proceed as in the previous case. We have that there exists $x \in N_H(v)$ non-adjacent to $z$, for otherwise the minimum dominating set $(S\setminus \{v\}) \cup \{z\}$ of $H$ has fewer bad vertices than $S$. Therefore $x$ is adjacent to $y$ and, since $H$ is $P$-free, also to $t$. Moreover, there exists $w \in N_H(u) \cup N_H(v)$ non-adjacent to both $x$ and $y$, for otherwise the minimum dominating set $(S\setminus \{u, v\}) \cup \{x, y\}$ of $H$ contains fewer bad vertices than $S$. 

\begin{figure}[htb]
\centering
\begin{tikzpicture}[node distance=1cm]
\node[circ,label=below:{\footnotesize $z$}] (v) {};
\node[circ,label=below:{\footnotesize $y$},right of=v] (z) {};  
\node[circ,label=below:{\footnotesize $u$},right of=z] (y) {};
\node[circ,label=above:{\footnotesize $v$},above of=v] (x) {};
\node[circ,label=above:{\footnotesize $x$},right of=x] (t) {};
\node[circ,label=above:{\footnotesize $t$},right of=t] (u) {};

\draw (v) -- (z)
(v) -- (x)
(z) -- (y)
(z) -- (t)
(y) -- (u)
(x) -- (t)
(t) -- (u);
\end{tikzpicture}
\caption{Case 2.}
\end{figure}

But then as in \labelcref{c2} and \labelcref{c1}, we have that $w$ is adjacent to both $t$ and $z$, and that each vertex in $(N_H(u) \cup N_H(v)) \setminus \{t, z\}$ is adjacent to either $t$ or $z$. Therefore, $(S\setminus \{u, v\}) \cup \{t, z\}$ is a minimum dominating set of $H$ with fewer bad vertices, a contradiction.
\end{proof}

\section{Concluding remarks and open problems}

In this paper, we continued the systematic study on \textsc{Semitotal Dominating Set} initiated in \citep{HP17}. In particular, we showed that there exists a $O(n^{4 + 6w})$ time algorithm for graphs with bounded mim-width, assuming a decomposition tree with mim-width $w$ is provided as part of the input. The idea consists in reducing the problem to \textsc{Total Dominating Set} via a graph transformation preserving this graph class and then use known algorithms for \textsc{Total Dominating Set}. We remark that there certainly is room for improvement in the $O(n^{4 + 6w})$ running time for specific graph classes with bounded mim-width. We also observed that the class of dually chordal graphs is closed under the same transformation and so a polynomial-time algorithm for \textsc{Total Dominating Set} for dually chordal graphs would provide a polynomial-time algorithm for \textsc{Semitotal Dominating Set} for that same class. The search for such an algorithm seems to be an interesting open problem.   

In \Cref{table} and \Cref{fig:dia}, we highlighted other graph classes for which the status of \textsc{Semitotal Dominating Set} is unknown. Most notably, AT-free graphs and its subclass of co-comparability graphs, circle graphs and tolerance graphs. The former three classes all have unbounded mim-width and it is an open problem to determine whether the mim-width of tolerance graphs is bounded or not \citep{Men17}. It is easy to see that the class of AT-free graphs is not closed under our transformation but it is not immediately clear whether the same holds for the other three classes.   

In \Cref{secapp}, we showed that \textsc{Semitotal Dominating Set} is $\mathsf{APX}$-complete when restricted to cubic graphs and to subcubic bipartite graphs. It is natural to expect that the same holds for cubic bipartite graphs and we leave the verification as an open problem. 

In \Cref{sec:dicho}, we obtained complexity dichotomies in monogenic classes for \textsc{Semitotal Dominating Set} and \textsc{Total Dominating Set} and observed that the complexities of \textsc{Dominating Set}, \textsc{Semitotal Dominating Set}, \textsc{Total Dominating Set} and \textsc{Connected Dominating Set} all agree when restricted to monogenic classes. 
Moreover, we asked the following: if \textsc{Dominating Set} is in $\mathsf{P}$ for a hereditary class $\mathcal{C}$, is it true that \textsc{Semitotal Dominating Set} is in $\mathsf{P}$ for $\mathcal{C}$ as well? In case of a negative answer, it would be interesting to investigate the following: if \textsc{Dominating Set} is polynomial for a class $\mathcal{C}$, can we improve on the trivial $2$-approximation algorithm for \textsc{Semitotal Dominating Set} restricted to $\mathcal{C}$?  

In \Cref{sec:perfect}, we showed that it is $\mathsf{NP}$-complete to recognise the graphs such that $\gamma_{t2}(G) = \gamma_{t}(G)$ and those such that $\gamma(G) = \gamma_{t2}(G)$, even if restricted to be planar and with maximum degree at most $4$. We conclude by asking what happens in the case of subcubic graphs.

\section*{Acknowledgments}

A. M. would like to thank Martin Milani\v{c} and Arne Leitert for useful discussions.
\bibliographystyle{plainnat}
\bibliography{references}

\begin{thebibliography}{55}
\providecommand{\natexlab}[1]{#1}
\providecommand{\url}[1]{\texttt{#1}}
\expandafter\ifx\csname urlstyle\endcsname\relax
  \providecommand{\doi}[1]{doi: #1}\else
  \providecommand{\doi}{doi: \begingroup \urlstyle{rm}\Url}\fi

\bibitem[AbouEisha et~al.(2018)AbouEisha, Hussain, Lozin, Monnot, Ries, and
  Zamaraev]{Abou18}
H.~AbouEisha, S.~Hussain, V.~Lozin, J.~Monnot, B.~Ries, and V.~Zamaraev.
\newblock Upper {D}omination: Towards a dichotomy through boundary properties.
\newblock \emph{Algorithmica}, 80\penalty0 (10):\penalty0 2799--2817, 2018.

\bibitem[Alekseev(2003)]{Ale03}
V.~E. Alekseev.
\newblock On easy and hard hereditary classes of graphs with respect to the
  {I}ndependent {S}et problem.
\newblock \emph{Discrete Applied Mathematics}, 132\penalty0 (1–3):\penalty0
  17--26, 2003.

\bibitem[Alimonti and Kann(2000)]{AK00}
P.~Alimonti and V.~Kann.
\newblock Some {APX}-completeness results for cubic graphs.
\newblock \emph{Theoretical Computer Science}, 237\penalty0 (1):\penalty0
  123--134, 2000.

\bibitem[Alvarado et~al.(2015)Alvarado, Dantas, and Rautenbach]{ADR15}
J.~D. Alvarado, S.~Dantas, and D.~Rautenbach.
\newblock Perfectly relating the domination, total domination, and paired
  domination numbers of a graph.
\newblock \emph{Discrete Mathematics}, 338\penalty0 (8):\penalty0 1424--1431,
  2015.

\bibitem[Arnborg et~al.(1991)Arnborg, Lagergren, and Seese]{ALS91}
S.~Arnborg, J.~Lagergren, and D.~Seese.
\newblock Easy problems for tree-decomposable graphs.
\newblock \emph{Journal of Algorithms}, 12\penalty0 (2):\penalty0 308--340,
  1991.

\bibitem[Belmonte and Vatshelle(2013)]{BV13}
R.~Belmonte and M.~Vatshelle.
\newblock Graph classes with structured neighborhoods and algorithmic
  applications.
\newblock \emph{Theoretical Computer Science}, 511:\penalty0 54--65, 2013.

\bibitem[Bertossi(1984)]{Ber84}
A.~A. Bertossi.
\newblock Dominating sets for split and bipartite graphs.
\newblock \emph{Information Processing Letters}, 19\penalty0 (1):\penalty0
  37--40, 1984.

\bibitem[Bodlaender(1996)]{Bod96}
H.~L. Bodlaender.
\newblock A linear-time algorithm for finding tree-decompositions of small
  treewidth.
\newblock \emph{SIAM Journal on Computing}, 25\penalty0 (6):\penalty0
  1305--1317, 1996.

\bibitem[Brandst\"{a}dt et~al.(1998{\natexlab{a}})Brandst\"{a}dt, Chepoi, and
  Dragan]{BCD98}
A.~Brandst\"{a}dt, V.~D. Chepoi, and F.~F. Dragan.
\newblock The algorithmic use of hypertree structure and maximum neighbourhood
  orderings.
\newblock \emph{Discrete Applied Mathematics}, 82\penalty0 (1):\penalty0
  43--77, 1998{\natexlab{a}}.

\bibitem[Brandst\"{a}dt et~al.(1998{\natexlab{b}})Brandst\"{a}dt, Dragan,
  Chepoi, and Voloshin]{BDCV98}
A.~Brandst\"{a}dt, F.~Dragan, V.~Chepoi, and V.~Voloshin.
\newblock Dually chordal graphs.
\newblock \emph{SIAM Journal on Discrete Mathematics}, 11\penalty0
  (3):\penalty0 437--455, 1998{\natexlab{b}}.

\bibitem[Brandst\"{a}dt et~al.(2015)Brandst\"{a}dt, Fi\v{c}ur, Leitert, and
  Milani\v{c}]{BFLM15}
A.~Brandst\"{a}dt, P.~Fi\v{c}ur, A.~Leitert, and M.~Milani\v{c}.
\newblock Polynomial-time algorithms for weighted efficient domination problems
  in {AT}-free graphs and dually chordal graphs.
\newblock \emph{Information Processing Letters}, 115\penalty0 (2):\penalty0
  256--262, 2015.

\bibitem[Bui-Xuan et~al.(2013)Bui-Xuan, Telle, and Vatshelle]{BTV13}
B.-M. Bui-Xuan, J.~A. Telle, and M.~Vatshelle.
\newblock Fast dynamic programming for locally checkable vertex subset and
  vertex partitioning problems.
\newblock \emph{Theoretical Computer Science}, 511:\penalty0 66--76, 2013.

\bibitem[Camby and Plein(2017)]{CP17}
E.~Camby and F.~Plein.
\newblock A note on an induced subgraph characterization of domination perfect
  graphs.
\newblock \emph{Discrete Applied Mathematics}, 217:\penalty0 711--717, 2017.

\bibitem[Chleb\'{i}k and Chleb\'{i}kov\'{a}(2008)]{CC08}
M.~Chleb\'{i}k and J.~Chleb\'{i}kov\'{a}.
\newblock Approximation hardness of dominating set problems in bounded degree
  graphs.
\newblock \emph{Information and Computation}, 206\penalty0 (11):\penalty0
  1264--1275, 2008.

\bibitem[Chudnovsky et~al.(2006)Chudnovsky, Robertson, Seymour, and
  Thomas]{CRST06}
M.~Chudnovsky, N.~Robertson, P.~Seymour, and R.~Thomas.
\newblock The strong perfect graph theorem.
\newblock \emph{Annals of Mathematics}, 164\penalty0 (1):\penalty0 51--229,
  2006.

\bibitem[Courcelle(1990)]{Cou90}
B.~Courcelle.
\newblock The monadic second-order logic of graphs. {I}. {R}ecognizable sets of
  finite graphs.
\newblock \emph{Information and Computation}, 85\penalty0 (1):\penalty0 12--75,
  1990.

\bibitem[Courcelle and Olariu(2000)]{CO00}
B.~Courcelle and S.~Olariu.
\newblock Upper bounds to the clique width of graphs.
\newblock \emph{Discrete Applied Mathematics}, 101\penalty0 (1–3):\penalty0
  77--114, 2000.

\bibitem[Courcelle et~al.(2000)Courcelle, Makowsky, and Rotics]{CMR00}
B.~Courcelle, J.~A. Makowsky, and U.~Rotics.
\newblock Linear time solvable optimization problems on graphs of bounded
  clique-width.
\newblock \emph{Theory of Computing Systems}, 33\penalty0 (2):\penalty0
  125--150, 2000.

\bibitem[Dahlhaus et~al.(1994)Dahlhaus, Johnson, Papadimitriou, Seymour, and
  Yannakakis]{DJPS94}
E.~Dahlhaus, D.~Johnson, C.~Papadimitriou, P.~Seymour, and M.~Yannakakis.
\newblock The complexity of multiterminal cuts.
\newblock \emph{SIAM Journal on Computing}, 23\penalty0 (4):\penalty0 864--894,
  1994.

\bibitem[Damaschke et~al.(1990)Damaschke, M\"{u}ller, and Kratsch]{DMK90}
P.~Damaschke, H.~M\"{u}ller, and D.~Kratsch.
\newblock Domination in convex and chordal bipartite graphs.
\newblock \emph{Information Processing Letters}, 36\penalty0 (5):\penalty0
  231--236, 1990.

\bibitem[Du et~al.(2012)Du, Ko, and Hu]{Du12}
Ding-Zhu Du, Ker-I Ko, and Xiaodong Hu.
\newblock \emph{Design and {A}nalysis of {A}pproximation {A}lgorithms}.
\newblock Springer Optimization and Its Applications. Springer, 2012.

\bibitem[Farber(1983)]{Far83}
M.~Farber.
\newblock Characterizations of strongly chordal graphs.
\newblock \emph{Discrete Mathematics}, 43\penalty0 (2):\penalty0 173--189,
  1983.

\bibitem[Garey and Johnson(1979)]{GJ79}
M.~R. Garey and D.~S. Johnson.
\newblock \emph{Computers and {I}ntractability: {A} Guide to the {T}heory of
  {NP}-Completeness}.
\newblock W. H. Freeman, 1979.

\bibitem[Giannopoulou and Mertzios(2016)]{GM16}
A.~Giannopoulou and G.~Mertzios.
\newblock New geometric representations and domination problems on tolerance
  and multitolerance graphs.
\newblock \emph{SIAM Journal on Discrete Mathematics}, 30\penalty0
  (3):\penalty0 1685--1725, 2016.

\bibitem[Goddard et~al.(2014)Goddard, Henning, and McPillan]{GHM14}
W.~Goddard, M.~A. Henning, and C.~A. McPillan.
\newblock Semitotal domination in graphs.
\newblock \emph{Utilitas Mathematica}, 94:\penalty0 67--81, 2014.

\bibitem[Golovach et~al.(2014)Golovach, Paulusma, and Song]{GPS14}
P.~A. Golovach, D.~Paulusma, and J.~Song.
\newblock Closing complexity gaps for coloring problems on {H}-free graphs.
\newblock \emph{Information and Computation}, 237:\penalty0 204--214, 2014.

\bibitem[Golumbic and Rotics(2000)]{GR00}
M.~C. Golumbic and U.~Rotics.
\newblock On the clique-width of some perfect graph classes.
\newblock \emph{International Journal of Foundations of Computer Science},
  11\penalty0 (03):\penalty0 423--443, 2000.

\bibitem[Haynes and Henning(2018)]{HH18}
T.~W. Haynes and M.~A. Henning.
\newblock Perfect graphs involving semitotal and semipaired domination.
\newblock \emph{Journal of Combinatorial Optimization}, 36\penalty0
  (2):\penalty0 416--433, 2018.

\bibitem[Haynes et~al.(1998{\natexlab{a}})Haynes, Hedetniemi, and
  Slater]{HHS98a}
T.~W. Haynes, S.~T. Hedetniemi, and P.~J. Slater.
\newblock \emph{Fundamentals of {D}omination in {G}raphs}.
\newblock Marcel Dekker Inc., 1998{\natexlab{a}}.

\bibitem[Haynes et~al.(1998{\natexlab{b}})Haynes, Hedetniemi, and
  Slater]{HHS98b}
T.~W. Haynes, S.~T. Hedetniemi, and P.~J. Slater, editors.
\newblock \emph{Domination in {G}raphs: {A}dvanced {T}opics,}.
\newblock Marcel Dekker Inc., 1998{\natexlab{b}}.

\bibitem[Henning and Pandey(2017)]{HP17}
M.~A. Henning and A.~Pandey.
\newblock Algorithmic aspects of semitotal domination in graphs.
\newblock \emph{CoRR}, abs/1711.10891, 2017.
\newblock URL \url{https://arxiv.org/pdf/1711.10891}.

\bibitem[Henning and Yeo(2013)]{HY13}
M.~A. Henning and A.~Yeo.
\newblock \emph{Total {D}omination in {G}raphs}.
\newblock Springer, 2013.

\bibitem[Jaffke et~al.(2018)Jaffke, Kwon, Str{\o}mme, and Telle]{JKST18}
L.~Jaffke, O.~Kwon, T.~J.~F. Str{\o}mme, and J.~A. Telle.
\newblock Generalized distance domination problems and their complexity on
  graphs of bounded mim-width.
\newblock \emph{CoRR}, 2018.
\newblock URL \url{https://arxiv.org/pdf/1803.03514.pdf}.

\bibitem[Kami\'{n}ski(2012)]{Kam12}
M.~Kami\'{n}ski.
\newblock {MAX-CUT} and containment relations in graphs.
\newblock \emph{Theoretical Computer Science}, 438:\penalty0 89--95, 2012.

\bibitem[Kami\'{n}ski et~al.(2009)Kami\'{n}ski, Lozin, and Milani\v{c}]{KLM09}
M.~Kami\'{n}ski, V.~V. Lozin, and M.~Milani\v{c}.
\newblock Recent developments on graphs of bounded clique-width.
\newblock \emph{Discrete Applied Mathematics}, 157\penalty0 (12):\penalty0
  2747--2761, 2009.

\bibitem[Keil(1993)]{Keil93}
J.~M. Keil.
\newblock The complexity of domination problems in circle graphs.
\newblock \emph{Discrete Applied Mathematics}, 42\penalty0 (1):\penalty0
  51--63, 1993.

\bibitem[Korobitsin(1992)]{Kor92}
D.~V. Korobitsin.
\newblock On the complexity of domination number determination in monogenic
  classes of graphs.
\newblock \emph{Discrete Mathematics and Applications}, 2\penalty0
  (2):\penalty0 191--200, 1992.

\bibitem[Kr\'{a}l' et~al.(2001)Kr\'{a}l', Kratochv\'{\i}l, Tuza, and
  Woeginger]{KKTW01}
D.~Kr\'{a}l', J.~Kratochv\'{\i}l, Z.~Tuza, and G.~J. Woeginger.
\newblock Complexity of coloring graphs without forbidden induced subgraphs.
\newblock In A.~Brandst\"{a}dt and V.~B. Le, editors, \emph{Graph-Theoretic
  Concepts in Computer Science}, volume 2204 of \emph{Lecture Notes in Computer
  Science}, pages 254--262. 2001.

\bibitem[Kratsch(2000)]{Kra00}
D.~Kratsch.
\newblock Domination and total domination on asteroidal triple-free graphs.
\newblock \emph{Discrete Applied Mathematics}, 99\penalty0 (1):\penalty0
  111--123, 2000.

\bibitem[Kratsch and Stewart(1997)]{KS97}
D.~Kratsch and L.~Stewart.
\newblock Total domination and transformation.
\newblock \emph{Information Processing Letters}, 63\penalty0 (3):\penalty0
  167--170, 1997.

\bibitem[Laskar and Pfaff(1983)]{LP83}
R.~Laskar and J.~Pfaff.
\newblock Domination and irredundance in split graphs.
\newblock Technical Report 430, Dept. Mathematical Sciences, Clemson Univ.,
  1983.

\bibitem[Leitert(2017)]{Lei17}
A.~Leitert.
\newblock $3$-colouring for dually chordal graphs and generalisations.
\newblock \emph{Information Processing Letters}, 128:\penalty0 21--26, 2017.

\bibitem[Malyshev and Pardalos(2015)]{MP15}
D.~S. Malyshev and P.~M. Pardalos.
\newblock Critical hereditary graph classes: a survey.
\newblock \emph{Optimization Letters}, pages 1--20, 2015.

\bibitem[McRae(1994)]{Mc94}
A.~A. McRae.
\newblock \emph{Generalizing $\mathsf{NP}$-completeness proofs for bipartite
  and chordal graphs}.
\newblock PhD thesis, Clemson University, 1994.

\bibitem[Mengel(2017)]{Men17}
S.~Mengel.
\newblock Lower bounds on the mim-width of some graph classes.
\newblock \emph{Discrete Applied Mathematics}, 2017.
\newblock \doi{https://doi.org/10.1016/j.dam.2017.04.043}.
\newblock URL
  \url{http://www.sciencedirect.com/science/article/pii/S0166218X17302330}.

\bibitem[M\"{u}ller and Brandst\"{a}dt(1987)]{MB87}
H.~M\"{u}ller and A.~Brandst\"{a}dt.
\newblock The $\mathsf{NP}$-completeness of {S}teiner {T}ree and {D}ominating
  {S}et for chordal bipartite graphs.
\newblock \emph{Theoretical Computer Science}, 53\penalty0 (2):\penalty0
  257--265, 1987.

\bibitem[Munaro(2017)]{Mun3}
A.~Munaro.
\newblock Boundary classes for graph problems involving non-local properties.
\newblock \emph{Theoretical Computer Science}, 692:\penalty0 46--71, 2017.

\bibitem[Oum and Seymour(2006)]{OS06}
S.~Oum and P.~Seymour.
\newblock Approximating clique-width and branch-width.
\newblock \emph{Journal of Combinatorial Theory, Series B}, 96\penalty0
  (4):\penalty0 514--528, 2006.

\bibitem[Pfaff et~al.(1983)Pfaff, Laskar, and Hedetniemi]{PLH83}
J.~Pfaff, R.~Laskar, and S.~T. Hedetniemi.
\newblock {NP}-completeness of {T}otal and {C}onnected {D}omination, and
  {I}rredundance for bipartite graphs.
\newblock Technical Report 428, Dept. Mathematical Sciences, Clemson Univ.,
  1983.

\bibitem[S{\ae}ther and Vatshelle(2016)]{SV16}
S.~H. S{\ae}ther and M.~Vatshelle.
\newblock Hardness of computing width parameters based on branch decompositions
  over the vertex set.
\newblock \emph{Theoretical Computer Science}, 615:\penalty0 120--125, 2016.

\bibitem[Schaudt(2012)]{Sch12}
O.~Schaudt.
\newblock On graphs for which the connected domination number is at most the
  total domination number.
\newblock \emph{Discrete Applied Mathematics}, 160\penalty0 (7):\penalty0
  1281--1284, 2012.

\bibitem[Telle and Proskurowski(1997)]{TP97}
J.~A. Telle and A.~Proskurowski.
\newblock Algorithms for vertex partitioning problems on partial $k$-trees.
\newblock \emph{SIAM Journal on Discrete Mathematics}, 10\penalty0
  (4):\penalty0 529--550, 1997.

\bibitem[Vatshelle(2012)]{Vat12}
M.~Vatshelle.
\newblock \emph{New {W}idth {P}arameters of {G}raphs}.
\newblock PhD thesis, University of Bergen, 2012.

\bibitem[Vazirani(2001)]{Vaz01}
V.~V. Vazirani.
\newblock \emph{Approximation {A}lgorithms}.
\newblock Springer-Verlag, 2001.

\bibitem[Zverovich(2003)]{Zve03}
Igor~E. Zverovich.
\newblock Perfect connected-dominant graphs.
\newblock \emph{Discussiones Mathematicae Graph Theory}, 23\penalty0
  (1):\penalty0 159--162, 2003.

\end{thebibliography}

\end{document}